\documentclass{article}
\usepackage[margin=1in]{geometry}

\usepackage[utf8]{inputenc}
\usepackage[T1]{fontenc}
\usepackage[utf8]{inputenc}
\usepackage{verbatim}
\usepackage{tikz}
\usetikzlibrary{shapes.geometric}
\usetikzlibrary{quotes}
\usetikzlibrary{arrows.meta}
\usepackage{subfig}

\usetikzlibrary{snakes,arrows,shapes}

\author{Aditya Jayaprakash \\
Department of Computing Science\\ 
University of Alberta
\and Mohammad R. Salavatipour\footnote{Supported by NSERC.}\\
Department of Computing Science\\ 
University of Alberta}
\date{}
\usepackage{hyperref}
\usepackage{bbm}
\usepackage{amsmath,amssymb,enumerate,ifthen,tikz,floatrow,amsmath, algorithm,algorithmic}
\usepackage{mathtools}
\usepackage{graphicx}

\graphicspath{ {./images/} }
\newtheorem{theorem}{Theorem} 
\newtheorem{lemma}{Lemma}

\newtheorem{definition}{Definition} 
\newtheorem{remarka}{Remark} 
\usepackage{float}

\def\polylog{\mathop{\rm polylog}\nolimits}
\usepackage{hyperref}
 
\newenvironment{proof}{{\bf Proof.}}{\hfill\rule{2mm}{2mm}} 
 
\newcommand{\calI}{{\cal I}}
\newcommand{\calT}{{\cal T}}

\newcommand{\Prob}[1]{\mathbb{P}{\left[#1\right]}}
\newcommand{\Ex}[1]{\mathbb{E}{\left[#1\right]}}

\newcommand{\opt}{\mbox{\sc opt}}
\newcommand{\OPT}{\mbox{\sc OPT}}

\newcommand{\RR}{\mathbb{R}}

\newcommand{\mi}[1]{\underset{#1}{\text{min }}}

\newcommand{\eps}{\epsilon}

\newcommand{\poly}{\text{poly}}

\DeclarePairedDelimiter\ceil{\lceil}{\rceil}

\newcommand{\cost}{\text{cost}}

\DeclarePairedDelimiter\autobracket{(}{)}
\newcommand{\br}[1]{\autobracket*{#1}}

\title{Approximation Schemes for Capacitated Vehicle Routing on Graphs of Bounded Treewidth, Bounded Doubling, or Highway Dimension}
\newcommand{\tvec}{\vec{t}}
\newcommand{\yvec}{\vec{y}}
\newcommand{\nvec}{\vec{n}}

\newcommand{\hvec}{\vec{h}}
\newcommand{\lvec}{\vec{l}}
\newcommand{\zvec}{\vec{z}}
\newcommand{\ovec}{\vec{o}}
\newcommand{\pvec}{\vec{p}}

\newcommand{\bbold}{\textbf{B}}
\newcommand{\cbold}{\textbf{C}}
\newcommand{\abold}{\textbf{A}}

\newcommand{\hbold}{\textbf{H}}
\newcommand{\ibold}{\textbf{I}}

\newcommand{\norm}[1]{\left\lVert#1\right\rVert}

\newcommand{\boundellipse}[3]
{(#1) ellipse (#2 and #3)
}
\begin{document}

\maketitle
\begin{abstract}
In this paper we present Approximation Schemes for Capacitated Vehicle Routing Problem (CVRP) on several classes of graphs.
In CVRP, introduced by Dantzig and Ramser in 1959 \cite{Dantzig}, we are given a graph $G=(V,E)$ with metric edges costs, a depot $r\in V$, and a vehicle of bounded capacity $Q$. The goal is
to find minimum cost collection of tours for the vehicle that return to the depot, each visiting at most $Q$ nodes, such that they 
cover all the nodes. This generalizes classic TSP and has been studied extensively. In the more general setting each node $v$ has
a demand $d_v$ and the total demand of each tour must be no more than $Q$. Either the demand of each node must be served by one tour
(unsplittable) or can be served by multiple tour (splittable). The best known approximation algorithm for general graphs has 
ratio $\alpha+2(1-\epsilon)$ (for the unsplittable) and $\alpha+1-\epsilon$ (for the splittable) for some fixed $\epsilon>\frac{1}{3000}$, where $\alpha$ is the best approximation for TSP.
Even for the case of trees, the best approximation ratio is $4/3$ \cite{Becker18} and it has been an open question if there is an approximation scheme for this simple class of graphs. Das and Mathieu \cite{Das-Mathieu} presented an approximation
scheme with time $n^{\log^{O(1/\epsilon)}n}$ for Euclidean plane $\RR^2$. No other approximation scheme is known for any other class of metrics
(without further restrictions on $Q$). In this paper we make significant progress on this classic problem by presenting Quasi-Polynomial Time Approximation Schemes (QPTAS) for graphs of bounded treewidth, graphs of bounded highway dimensions, and graphs of bounded doubling dimensions. For comparison, our result implies an approximation scheme for Euclidean plane with run time $n^{O(\log^{10}n/\eps^{9})}$. 
\end{abstract}

\section{Introduction}
Vehicle routing problems (VRP) describe a class of problems where the objective is to find cost efficient delivery routes for delivering items from depots to clients using vehicles having limited capacity. These problems have numerous applications in real world settings. The Capacitated Vehicle Routing Problem (CVRP) was introduced by Dantzig and Ramser in 1959 \cite{Dantzig}. In CVRP, we are given as input a graph $G=(V,E)$ with
metric edge weights (also referred to as costs) $w(e)\in\mathbb{Z}^{\ge 0}$,
a depot $r\in V$, along with a vehicle of capacity $Q>0$, and wish to compute a minimum weight/cost collection of tours, each starting from the depot and visiting at most $Q$ customers, whose union covers all the customers. In the more general setting each node $v$ has a demand $d(v)\in \mathbb{Z}^{\ge 1}$ and the goal is to find a set of tours of minimum total cost each of which includes $r$ such that the union of the tours covers the demand at every client and every tour covers at most $Q$ demand. 

There are three common versions of CVRP: \emph{unit}, \emph{splittable}, and \emph{unsplittable}. In the splittable variant, the demand of a node can be delivered using multiple tours, but in the unsplittable variant, the entire demand of a client must be delivered by a single tour. The unit demand case is a special case of the unsplittable case where every node has unit demand and the demand of a client must be delivered by a single tour. CVRP has also been referred to as the $k$-tours problem \cite{Arora-Euclidean-PTAS, stoc/AsanoKTT97}.
 All three variants admit constant factor approximation algorithm in polynomial-time \cite{Haimovich-Kan}. Haimovich et al. \cite{Haimovich-Kan} showed that a heuristic
 called iterative partitioning (which starts from a TSP tour and breaking the tour into capacity respecting tours by making a trip back and  forth to the depot) implies an $(\alpha+1(1-1/Q))$-approximation
 for the unit demand case, with $\alpha$ being the approximation ratio of Traveling Salesman Problem (TSP). 
 A similar approach implies a
$2 + (1 - 2/Q)\alpha)$-approximation for the unsplittable variant \cite{ALTINKEMER1987149}. Very recently, Blauth et al. \cite{Vygen} improved these approximations by showing that there is an $\eps > 0$ such that there is an $(\alpha + 2 \cdot (1 - \eps))$-approximation algorithm for unsplittable CVRP and a $(\alpha + 1 - \eps)$-approximation algorithm for unit demand CVRP and splittable CVRP. For $\alpha = 3/2$, they showed $\eps > 1/3000$. 
 All three variants are APX-hard in general metric spaces \cite{Papadimitrio-Yannakakis}, so a natural research focus has been on structured metric spaces, i.e. special graph classes. Even on on trees (and in particular on stars) CVRP remains NP-hard \cite{Labbe-Mercure}, and there exists constant-factor approximations (currently being $4/3$ \cite{Becker18}), better than those for general metrics, however the following question has remained open:\\
\textbf{Question.} Is it possible to design an approximation scheme for CVRP on trees or more generally graphs of bounded treewidth? 


We answer the above question affirmatively. For ease of exposition we start by prove the following first:
\begin{theorem}\label{thm:tree}
For any $\eps > 0$, there is an algorithm that, for any instance of the unit demand CVRP on trees outputs a $(1 + \eps)$-approximate solution in time $n^{O(\log^4 n/\eps^3)}$.
For any instance of the splittable CVRP on trees when $Q = n^{O(\log^c n)}$ the algorithm  
runs in time $n^{O(\log^{2c + 4}n)}$.
\end{theorem}



We then show how this result can be extended to design QPTAS for graphs of bounded treewidth.

\begin{theorem}\label{thm:treewidth}
For any $\eps > 0$, there is an algorithm that, for any instance of the unit demand CVRP on a graph $G$ of bounded treewidth $k$ outputs a $(1+\eps)$-approximate solution in time 
$n^{O(k^2\log^3 n/\eps^2)}$. For the splittable CVRP on graphs of bounded treewidth when $Q = n^{O(\log^c n)}$, 
the algorithm outputs a $(1 + \eps)$-approximate solution in time $n^{O(k^2\log^{2c + 3} n/\eps^2)}$.
\end{theorem}



As a consequence of this and using earlier results of embedding of graphs of bounded doubling dimensions or bounded highway dimensions into graphs of low treewidth we obtain approximation schemes for CVRP on those graph classes.

\begin{theorem}\label{thm:DD}
For any $\eps > 0$ and fixed $D > 0$, there is a an algorithm that, given an instance of the splittable CVRP with 
capacity $Q = n^{\log^c n}$ on a graph of doubling dimension $D$, finds a $(1 + \eps)$-approximate solution in time $n^{O(D^D \log^{2c + D + 3}n/\eps^{D+2})}$.
\end{theorem}

As an immediate corollary, this implies an approximation scheme for CVRP on Euclidean metrics on $\mathbb{R}^2$ in time $n^{O(\log^{10}n/\eps^{9})}$
which improves on the run time of $n^{\log^{O(1/\epsilon)}n}$ of QPTAS of \cite{Das-Mathieu}. 

\begin{theorem}\label{thm:HD}
For any $\eps > 0, \lambda > 0$ and $D > 0$, there is a an algorithm that, given a graph with highway dimension $D$ with violation $\lambda$ as an instance of the splittable CVRP with capacity $Q = n^{\log^c n}$, finds a solution whose cost is at most $(1 + \eps)$ times the optimum in time $n^{O( \log^{2c + 3 + \log^2(\frac{D}{\eps \lambda})\cdot \frac{1}{\lambda}}n/\eps^2)}$.
\end{theorem}

\subsection{Related Works}
CVRP generalizes the classic TSP problem (with $Q=n$).
For general metrics, Haimovich et al. \cite{Haimovich-Kan} considered a simple heuristic, called tour partitioning, which starts from a TSP tour and then splits the tour into tours of size at most $Q$ (by making back-and-forth trips to $r$) and showed that it is a $(1 + (1 - 1/Q)\alpha)$-approximation for splittable CVRP, where $\alpha$ is the approximation ratio for TSP. Essentially the same algorithm implies a $(2 + (1 - 2/Q)\alpha)$-approximation for unsplittable CVRP \cite{ALTINKEMER1987149}. These stood as the best known bounds until recently, when Blauth et al. \cite{Vygen} showed that given a TSP approximation $\alpha$, there is an $\eps > 0$ such that there is an $(\alpha + 2 \cdot (1 - \eps))$-approximation algorithm for CVRP. For $\alpha = 3/2$, they showed $\eps > 1/3000$. They also showed a $(\alpha + 1 - \eps)$-approximation algorithm for unit demand CVRP and splittable CVRP. 

For the case of trees, Labbé et al. \cite{Labbe-Mercure} showed splittable CVRP is NP-hard and Golden et al. \cite{Golden-Wong} showed unsplittable version is APX-hard and hard to approximate better than 1.5. For splittable CVRP (again on trees), Hamaguchi et al. \cite{Hamaguchi-Katoh} defined a lower bound for the cost of the optimal solution and gave a 1.5 approximation with respect to the lower bound. Asano et al. \cite{stoc/AsanoKTT97} improved the approximation to $(\sqrt{41} - 1)/4$ with respect to the same lower bound and also showed the existence of instances whose optimal cost is exactly 4/3 times the lower bound. Becker \cite{Becker18} gave a 4/3-approximation with respect to the  lower bound. Becker and Paul \cite{Becker-Paul-Bricriteria} showed a $(1, 1+ \eps)$-bicriteria polynomial-time approximation scheme for splittable CVRP in trees, i.e. a PTAS but the capacity of every tour is up to $(1+\eps)Q$.

Das and Mathieu \cite{Das-Mathieu} gave a quasi-polynomial-time approximation scheme (QPTAS) for CVRP in the Euclidean plane ($\mathbb{R}^2$). A PTAS for when $Q$ is $O(\log n/\log \log n)$ or $Q$ is $\Omega(n)$ was shown by Asano et al. \cite{stoc/AsanoKTT97}. A PTAS for Euclidean plane $\mathbb{R}^2$ for all moderately large values of $Q \le 2^{\log^\delta n}$, where $\delta = \delta(\eps)$, was shown by Adamaszek et al \cite{AdamaszekCL09}, building on the work of Das and Mathieu \cite{Das-Mathieu}, and using it as a subroutine. For high dimensional Euclidean spaces $\mathbb{R}^d$, Khachay et al. \cite{Khachay-PTAS} showed a PTAS when $Q$ is $O(\log^{1/d}n)$. For graphs of bounded doubling dimension, Khachay et al. \cite{Khachay-moderatenumer} gave a QPTAS when the number of tours is $\polylog(n)$ and Khachay et al.  \cite{Khachay-moderatecapacity} gave a QPTAS when $Q$ is $\polylog(n)$.

The following results are all for when $Q$ is a fixed. CVRP is APX-hard in general metrics and is polynomial-time solvable on trees. There exists a PTAS for CVRP in the Euclidean plane ($\mathbb{R}^2$) (again for when $Q$ is fixed) as shown by Khachay et al. \cite{Khachay-PTAS}. A PTAS for planar graphs was shown by Becker et al. \cite{PlanarPTAS-Klein} and  a QPTAS for planar and bounded-genus graphs was shown by Becker et al. \cite{QPTAS-Planar-boundedgenus}.  A PTAS for graphs of bounded highway dimension and an exact algorithm for graphs with treewidth with running time $O(n^{\text{tw}Q})$ was shown by Becker et al \cite{Becker-boundedhighway}.  Cohen-Addad et al. \cite{Klein-minorfree} showed an efficient PTAS for graphs of bounded-treewidth, an efficient PTAS for bounded highway dimension, an efficient PTAS for bounded genus metrics and a QPTAS for minor-free metrics. Again, note that these results are
all under the assumption that $Q$ is fixed.

So aside from the QPTAS of \cite{Das-Mathieu} for $\mathbb{R}^2$ and subsequent slight generalization of \cite{AdamaszekCL09} no approximation scheme is known for CVRP on any non-trivial metrics for arbitrary values of $Q$ (even for trees). Standard ways of
extending a dynamic programs for Euclidean metrics to bounded doubling metrics do not seem to work to extend the results of
\cite{Das-Mathieu} to doubling metrics in quasi-polynomial time.

\subsection{Overview of our technique}
We start by presenting a QPTAS for CVRP on trees and then extend the technique to graphs of bounded
treewidth.
Our main technique to design approximation scheme for CVRP is to show the existence of a near optimum solution where the sizes of the partial tours going past any node of the tree can be partitioned into
only poly-logarithmic many classes. This will allow one to use dynamic programming to find a low cost solution. A simple rounding of tour sizes to some threshold values (e.g. powers of $(1+\eps)$) only works (with some care) to achieve a
bi-criteria approximation as any under estimation of tour sizes may result in tours that are violating the capacities. To achieve a true approximation (without capacity violation) we show how we can break the tours of an optimum solution into "top" and "bottom" parts (at any node $v$) and then swap the bottom parts of tours with the bottom parts of other tours which are smaller, and then "round them up" to the nearest value from a set of poly-logarithmic threshold values. This swapping creates enough room to do the "round up" without violating the capacities. However, this will cause a small fraction of the vertices to become "not covered", we call them  orphant nodes. We will show how we can randomly choose some tours of the optimum and add them back to the solution (at a small extra cost) and use these extra tours (after some modifications) to cover the orphant nodes. There are many details along the way. For instance, we treat the demand of each node as a token to be picked up by a tour. To ensure partial
tour sizes are always from a small (i.e. poly-logarithmic) size set, we add extra tokens over the nodes. 
Also, for our QPTAS to work we need to bound the height of the tree. We show how we can reduce the height of the tree to poly-logarithmic at a small loss using a height reduction lemma that might prove useful for other vehicle routing problems.

The technique of QPTAS for trees then can be extended to graphs of bounded treewidth and also graphs of bounded doubling dimension; prove the existence of a similar near optimum solution and find one using dynamic program. Or one can use the known results for embedding of graphs of bounded doubling dimension into graphs of small treewidth.

\section{Preliminaries}
Recall that an instance $\calI$ to CVRP is a graph $G = (V,E)$, where $w(e)$ is the cost or weight of edge $e \in E$ and $Q$ is the capacity of the vehicle. Each tour $\calT$ is a walk over some nodes of $G$. We say $\cal T$ "covers" node $v$ if it serves the demand at node $v$. For the unit demand CVRP, it is easier to think of the demand of each node $v$ as being a token on $v$ that must be picked up by a tour. We can generalize this and assume each node $v$ can have multiple tokens and the total number of tokens a tour can pick is most $Q$ (possibly
from the same or different locations). Note that each tour might visit vertices without picking any token there. 
The goal is to find a collection of tours of minimum total cost such that each token is picked up (or say covered) by some tour. 
We use $\OPT(G)$ or simply $\OPT$ to refer to an optimum solution of $G$, and $\opt$ to denote the value of it. Fix an optimal solution $\OPT$. For any edge $e$ let $f(e)$ denote the number of tours travelling edge $e$ in $\OPT$;
so $\opt=\sum_e w(e)\cdot f(e)$.

 First we show the demand of each node is bounded by a function of $Q$. And then,
using standard scaling and rounding and at a small loss, we show we can assume the edge weights are polynomially bounded (in $n$).   Given an instance for splittable CVRP with $n$ nodes and capacity $Q$, it is possible that the demand $d(v) > Q$ for some node $v$. From the work of Adamaszek et al \cite{AdamaszekCL09}, we will show how we can assume that the demand at each node $v$ satisfies $1 \le d(v) < nQ$. Adamaszek et al \cite{AdamaszekCL09} defined a \emph{trivial} tour to be a tour which picks up tokens from a single node in $T$ and a tour is \emph{non-trivial} if the tour picks up tokens from at least two nodes in $T$. They defined a \emph{cycle} to be a set of tours $t_1,\ldots, t_m (m \ge 2)$ and a set of nodes $\ell_1, \ell_2, \ldots, \ell_m, \ell_{m+1} = \ell_1$ such that each tour $t_i$ covers locations $\ell_i$ and $\ell_{i+1}$ and the origin is not considered as a node in $\ell_1, \ldots, \ell_m$.  They showed in Lemma 1 of \cite{AdamaszekCL09} that there is an optimal solution in which there are no cycles. Since there are no 2-cycles, there are no two tours which cover the same pair of nodes. So there is an optimal solution such that there are at most $n$ non-trivial tours (as argued in \cite{AdamaszekCL09}). So putting aside trivial tours (each picking up $Q$ tokens at a node), we can assume we have a total of at most $nQ$ tokens and in particular each node has at most this many tokens. Without loss of generality, we assume we have removed all trivial tours and so there is a total of at most $nQ$ demands.
  
  We can also assume there is at most one tour in $\OPT$ covering at most $Q/2$ demand. If there are at least two tours $\calT_1$ and $\calT_2$ covering less than $Q/2$ demand, they can be merged into a single tour at no additional cost. Since the total demand is at most $nQ$, the total number of tours in the optimal solution is at most $nQ/(Q/2) = 2n$. 
  
Now we scale edge weights to be polynomially bounded.
Observe that each tour in $\OPT$ traverses each edge $e$ at most once in each direction, so at most twice. 
Suppose we have guessed the largest edge weight that belongs to $\OPT$ (by enumerating over all possible such guesses) and have removed any edge with weight larger. Let $W=\max_{e\in E} w(e)$ be
the largest (guessed) edge in $\OPT$.
Suppose we build instance $\calI'$ by rounding up the weight of each edge $e$ to be maximum of
$w(e)$ and $\eps W/4n^3$. Since there are a total of at most $2n$ tours in $\OPT$ and each edge
is traversed at most twice by each tour, and there are at most $n^2$ edges,
the cost of solution $\OPT$ in $\calI'$ is at most
$\opt+4n\cdot n^2\cdot\frac{\eps W}{4n^3} \leq (1+\eps )\opt$. Note that the ratio of maximum to minimum edge weight in $\calI'$ is $4n^3/\eps$, but the edge weights are not necessarily integer.
Now suppose we scale the edge weights so that the minimum edge weight is 1 and the maximum edge weight is $4n^3/\eps$ and
then scale them all by $1/\eps$, and then round each one up to the nearest integer.
Note that by this rounding to the nearest integer, 
the cost of each edge is increased by a factor of at most $1+\eps$, so the cost of an optimum solution in
the new instance is at most $(1+\eps)(1+\eps)=(1+O(\eps))$ factor larger than before rounding while the edge weights are all polynomially bounded integers.
So from now on we assume we have this property for the given instance at a small loss.

We will use the following two simplified version of the Chernoff Bound \cite{Mitzenmacher} in our analysis. 
\begin{lemma}[\bf Chernoff bound]
Let $Y = \sum_{i = 1}^n Y_i$ where $Y_i = 1$ with probability $p_i$ and 0 with probability $1- p_i$, and all $Y_i$'s are independent. With $\mu = \Ex{Y}$,  $\Prob{Y > 2\mu} \le e^{-\mu/3}$
and $\Prob{Y < \frac{\mu}{2}} \le e^{-\mu/8}.$
\end{lemma}

\section{QPTAS for CVRP on Trees}
In this section we prove Theorem \ref{thm:tree}.
We will first prove a structure theorem which describes structural properties of a near-optimal solution. We will leverage these structural properties and use dynamic programming to compute a near-optimal solution. 

\subsection{Structure Theorem}

Our goal in this section is to show the existence of a near optimum solution (i.e. one with cost $(1 + O(\eps))\opt$) with
certain properties which makes it easy to find one using dynamic programming.
More specifically we show we can modify the instance $\calI$ to instance $\calI'$ on the same tree $T$ where each node has $\geq 1$
tokens (so possibly more than 1) and change $\OPT$ to a solution $\OPT'$ on $\calI'$ where cost of $\OPT'$ is at most
$(1+O(\eps))\opt$. Clearly the tours of $\OPT'$ form a capacity respecting solution of $\calI$ as well (of no more cost).

A starting point in our structure theorem is to show that given input tree $T$, for any $\epsilon>0$, we can build another tree $T'$ of height $O(\log^2 n/\epsilon)$ such that the cost of an
optimum solution in $T'$ is within $1+\epsilon$ factor of the optimum solution to $T$. We can lift a near-optimum solution to $T'$ into a near-optimum solution of $T$. We will show the following in Subsection \ref{sec:heightreduction}

\begin{theorem}\label{thm:height-red}
Given a tree $T$ as an instance of CVRP and for any fixed $\epsilon>0$, one can build a tree $T'$ with height $\delta\log^2 n/\epsilon$, for some fixed $\delta>0$, such that $\opt(T')\leq \opt(T)\leq (1+\epsilon)\opt(T')$.
\end{theorem}

So for the rest of this section we assume our input tree has height $O(\log^2 n/\eps)$ at a loss of (yet another) $1+\eps$ in 
approximation ratio.

\subsubsection{Overview of the ideas}
Let us give a high level idea of the Structure theorem. In order to do that it is helpful to start from a simpler task of developing a bi-criteria approximation scheme\footnote{Note that \cite{Becker-Paul-Bricriteria} already presents a bicriteria PTAS for CVRP on trees.
We present a simple bi-criteria QPTAS here as it is our starting point towards a true approximation scheme.}

Let $\calT$ be a tour in $\OPT$ and $v$ be a node in $T$. 
The \textbf{coverage} of $\calT$ with respect to $v$ is the number of tokens picked by $\calT$ in the subtree $T_v$.

Suppose a tour $\calT$ visits node $v$. We refer to the subtour of $\calT$ in $T_v$ (subtree rooted at $v$) as a partial tour.


{\bf A Bicriteria QPTAS:} For simplicity, assume $T$ is binary (this is not crucial in the design of the DP). 
A subproblem would be based on a node $v\in T$ and the structure of partial tours going into $T_v$ to pick up tokens in $T_v$
at minimum cost.
In other words, if one looks at the sections of tours of an optimum solution that cover tokens of $T_v$, what are the capacity profiles of those sections? For a vector $\tvec$ with $Q$ entries, where $\tvec_i$ (for each $1 \le i \le Q)$ is the number of partial tours going down $T_v$ which pick $i$ tokens (or their capacity for that portion is $i$), entry $\abold[v, \tvec]$ would store the minimum cost of covering $T_v$ with (partial) tours whose capacity profile is given by $\tvec$. It is not hard to fill this table's entries using a simple recursion based on the entries of children of $v$. So one can solve the CVRP problem "exactly" in time $O(n^{Q+1})$. We can reduce the time complexity by storing "approximate" sizes of the partial tours for each $T_v$.
So let us "round" the capacities of the tours into $O(\log Q/\epsilon)$  buckets, where bucket $i$ represents capacities that are in
$[(1+\eps)^{i-1},(1+\eps)^i)$. More precisely, consider threshold-sizes $S = \{\sigma_1, \ldots, \sigma_\tau \}$ where: for $1\leq i\leq 1/\epsilon$, $\sigma_i=i$, and for each value $i>1/\epsilon$: $\sigma_i = \sigma_{i-1}(1 + \eps)$ and $\sigma_\tau = Q$. Note 
that $|S|=O(\log Q/\epsilon)=O(\log n/\eps)$. Suppose we allow each tour to pick up to $(1+\epsilon)Q$
tokens. If it was the case that each partial tour for $T_v$ (i.e. part of a tour that enters/exits $T_v$) has a capacity that is also threshold-size (this may not be true!) then the DP table entries
would be based on vectors $\tvec$ of size $O(\log n/\epsilon)$, and the run time would be quasi-polynomial. One has to note that
for each subproblem of the optimum at a node $v$ with children $u,w$, even if the tour sizes going down $T_v$
were of threshold-sizes, the partial tours at $T_u$ and $T_w$ do not necessarily satisfy this property.

To extend this to a proper bicriteria $(1+\epsilon)$-approximation we can define the thresholds based on powers of 
$1+\epsilon'$ where $\epsilon'=\frac{\epsilon^2}{\log^2 n}$ instead: let 
$S = \{\sigma_1, \ldots, \sigma_\tau \}$ where $\sigma_i=i$ for $1\leq i\leq 1/\epsilon'$, and 
for $i>1/\epsilon'$ we have $\sigma_i = \sigma_{i-1}(1 + \eps')$, and $\sigma_\tau = Q$. So now 
$|S|=O(\log^2 n\cdot\log Q/\epsilon)=O(\log^3 n/\epsilon^2)$ when $Q = \poly(n)$. For each vector $\tvec$ of size $\tau$, where $0\leq t_i\leq n$
is the number of partial tours with coverage/capacity $\sigma_i$, let $A[v,\tvec]$ store the minimum cost of a collection of
(partial) tours covering all the tokens in $T_v$ whose capacity profile is $\tvec$, i.e. the number of
tours of size in $[\sigma_i,\sigma_{i+1})$ is $\tvec_i$. To compute the solution for $A[v,\tvec]$, 
given all the solutions for its two children $u,w$ we can do the following: consider two partial solutions,  
$A[u,\tvec_u]$ and $A[w,\tvec_w]$. One can combine some partial tours of $A[u,\tvec_u]$ with some partial tours of $A[w,\tvec_w]$,
i.e. if ${\cal T}_u$  is a (partial) tour of class $i$ for $T_u$ and ${\cal T}_w$ is a partial tour of class $j$ for $T_w$
then either these two tours are in fact part of the same tour for $T_v$, or not. In the former case, the partial tour
for $T_v$ obtained by the combination of the two tours will have cost $w({\cal T}_u)+w({\cal T}_w)+2w(vu)+2w(vw)$ and capacity
$t_i+t_j$ (or possibly $t_i+t_j+1$ if this tour is to cover $v$ as well). 
In the latter case, each of ${\cal T}_u$ and ${\cal T}_w$ extend (by adding edges $vu$ and $vw$, respectively) into
partial tours for $T_v$ of weights $w({\cal  T}_u)+2w(vu)$ and $w({\cal  T}_w)+2w(vw)$ (respectively)  and capacities $t_i$
and $t_j$ (or perhaps $t_i+1$ or $t_j+1$ if one of them is to cover $v$ as well). In the former case, since $t_i+t_j$ is not
a threshold-size, we can round it (down) to the nearest threshold-size. We say partial solutions for $T_v$, $T_u$ and $T_w$
are consistent if one can obtain the partial solution for $T_v$ by combining the solutions for $T_v$ and $T_w$. Given
$A[v,\tvec]$, we consider all possible subproblems $A[u,\tvec_u]$ and $A[w,\tvec_w]$ that are consistent and take the minimum cost among
all possible ways to combine them to compute $A[v,\tvec]$. 
Note that whenever we combine two solutions, we might be rounding
the partial tour sizes down to a threshold-size, so we "under-estimate" the actual tour size by a factor
of $1+\epsilon'$ in each subproblem calculation. Since the height of the tree is $h=O(\log^2 n/\epsilon)$, the actual error in the tour sizes computed at the root is at most $(1+\epsilon')^h=(1+O(\epsilon))$, so each tour will have size at most $(1+O(\epsilon))Q$. The time to compute each entry $A[v,\tvec]$ can be upper bounded
by $n^{O(\log^3 n/\epsilon^2)}$ and since there are $n^{O(\log^3 n/\epsilon^2)}$ subproblems, 
the total running time of the algorithm will be $n^{O(\log^3 n/\epsilon^2)}$. We can handle the setting where the tree is not binary  (i.e. each node $v$ has more than two children) by doing an inner DP, like a knapsack problem over children of $v$
(we skip the details here as we will explain the details for the actual QPTAS instead). 

\textbf{Going from a Bicriteria to a true QPTAS:}
Our main tool to obtain a true approximation scheme for CVRP in trees is to show the existence of a near-optimum solution where
the partial solutions for each $T_v$ have sizes that can be grouped into polyogarithmic many buckets as in the case of bi-criteria solution. 
Roughly speaking, starting from an optimum solution $\OPT$, we follow a bottom-up scheme and modify $\OPT$ by changing the solution at each $T_v$: 
at each node $v$, we change the structure of the tours going down $T_v$ (by adding a few extra tours from the depot) and
also adding some extra tokens at $v$ so that the
partial tours that visit $T_v$ all have a size from one of polyogarithmic many possible sizes (buckets)
while increasing the number and the cost of the tours by a small factor. We do this by duplicating some of the tours
that visit $T_v$ while changing parts of them that go down in $T_v$ and adding some extra tokens at $v$:
each tour still picks up at most a total of $Q$ tokens and the size (i.e. the number of tokens picked)
for each partial tour in the subtree $T_v$ is one of $O(\log^4 n/\eps^2)$ many possible values, while the total cost
of the solution is at most $(1+O(\eps))\opt$.

Suppose $T$ has height $h$ (where $h=\delta\log^2 n/\epsilon$). 
Let $V_\ell$ (for $1\leq \ell \leq h$) be the set of vertices at level
$\ell$ of the tree where $V_1=\{r\}$ and for each $\ell\geq 2$, $V_\ell$ are those vertices whose parent is in level $\ell-1$. For every tour $\calT$ and every level $\ell$, the top part of $\calT$ w.r.t. $\ell$ (denoted by $\calT_\ell^{top}$), is the part of $\calT$ induced by the vertices in $V_1\cup\ldots\cup V_{\ell-1}$ and the bottom part of $\calT$ are the partial tours of $\calT$ in the subtrees rooted at a vertex in $V_\ell$. Note that if we replace each partial tour of
the bottom part of a tour $\calT$ with a partial tour of a smaller capacity, the tour remains a capacity respecting tour.
Consider a node $v$ (which is at some level $\ell$) and suppose we have $n_v$ partial tours covering $T_v$. Let the $n_v$ tours in increasing order of their coverage be $t_{1}, \ldots, t_{n_v}$. Let $|t_i|$ be the coverage of tour $t_i$ (so $|t_i| \le |t_{i+1}|$). 
For a $g$ (to be specified later), we add enough empty tours to the beginning of this list so that the number of tours is a multiple of $g$. Then, we will put these tours into groups $G^v_{1}, \ldots, G^v_{g}$ of equal sizes by placing
the $i$'th $ n_v/g $ partial tours into $G^v_i$.
Let $h^{v,max}_i$ ($h^{v,min}_i$) refer to the maximum (minimum) size of the tours in $G^v_{i}$. This grouping is similar to the
grouping in the asymptotic PTAS for the classic bin-packing problem.
Note that $h^{v,max}_{i} \le h^{v,min}_{i+1}$. 

Consider a mapping $f$ where it maps each partial tour in $G^v_i$ to one in $G^v_{i-1}$ in the same order, i.e. the largest partial tour in $G^v_{i}$ is mapped to the largest in $G^v_{i-1}$, the 2nd largest to the 2nd largest and so on, for $i>1$ (suppose $f(.)$ maps all the tours of $G^v_1$ to empty tours). 
Now suppose we modify $\OPT$ to  $\OPT'$ in the following way: for each tour $\calT$ that has a partial
tour $t\in G^v_i$, replace the bottom part of $\calT$ at $v$ from $t$ to $f(t)$ (which is in $G^v_{i-1}$).
Note that by this change, the size of any tour like $\calT$ can only decrease.
Also, if instead of $f(t)$ we had replaced $t$ with a partial
tour of size $h^{v,max}_{i-1}$, it would still form a capacity respecting solution with the rest of $\calT$, 
because $h^{v,max}_{i-1}\leq h^{v,min}_i\leq |t|$. The only problem is that those tokens in $T_v$ that were picked
by the partial tours in $G^v_g$ are not covered by any tours; we call these {\em orphant} tokens. 
For now, assume that we add a few extra tours to $\OPT$ at low cost such that they cover all the orphant tokens of $T_v$.
If we have done this change for all vertices $v\in V_\ell$, then for every  tour like $\calT$,
the partial tours of $\calT$ going down each $T_v$ (for $v\in V_\ell$) are replaced with partial tours from a group 
one index smaller.
This means that, after these changes, 
for each tour $\calT$ and its (new) partial tour $t\in G^v_i$, if we add $h^{v,max}_i-|t|$ extra tokens at $v$ to be picked up by $t$ then each partial tour has size exactly the same as the maximum size of its group without violating the capacities.
This helps us store a compact "sketch" 
for partial solutions
at each node $v$ with the property that the partial solution can be extended to a near optimum one.

How to handle the case of orphant tokens (those picked by the tours in the the last groups $G^v_g$ before the swap)? 
We will show that if $n_v$ is sufficiently large (at least polylogarithmic) then if we sample a small fraction of the tours of the optimum at random and add two copies of them (as extra tours), they
can be used to cover the orphant tokens. So overall, we show how one can modify $\OPT$ by adding some extra tours to it
at a cost of at most $\eps\cdot\opt$ such that: each node $v$ has $\geq 1$ tokens and 
the sketch of the partial tours at each node $v$ is compact (only polyogarithmic many possible sizes) while the dropped tokens overall can be covered by the extra tours.

\subsubsection{Changing $\OPT$ to a near optimum structured solution}
We will show how to modify the optimal solution $\OPT$ to a near-optimum solution $\OPT'$ for a new instance $\calI'$
which has $\geq 1$ token at each node with certain properties. 
We start from $\ell=h$ and let $\OPT'=\OPT_\ell=\OPT$ and for decreasing values of $\ell$, we will show how
to modify $\OPT_{\ell+1}$ to
obtain $\OPT_{\ell}$. To obtain $\OPT_{\ell}$ from $\OPT_{\ell+1}$ we keep the partial tours at levels $\geq \ell$ the same as
$\OPT_{\ell+1}$ but we change the top parts of the tours and how the top parts can be matched to the partial tours at 
level $\ell$ so that together they form capacity respecting solutions (tours of capacity at most $Q$) at low cost.

First, we assume that $\OPT$ has at least $d\log n$ many tours for some sufficiently large $d$. Otherwise,
if there are at most $D=d\log n$ many tours in $\OPT$ we can do a simple DP to compute $\OPT$:
for each node $v$, we have a sub problem $A[v,T^v_1,\ldots,T^v_D]$ which stores the minimum cost solution
if $T^v_i$ is the number of vertices the $i$'th tour is covering in the subtree $T_v$. It is easy to fill this table in
time $O(n^D)$ having computed the solutions for its children.

\begin{definition}
Let \textbf{threshold values} be $\{\sigma_1, \ldots, \sigma_\tau \}$ where $\sigma_i=i$ for $1\leq i\leq \lceil 1/\epsilon\rceil$,
 and for $i>\lceil 1/\epsilon\rceil$ we have $\sigma_i = \lceil\sigma_{i-1}(1 + \eps)\rceil$, and $\sigma_\tau = Q$. 
So $\tau=O(\log Q/\epsilon)$.
\end{definition}

We consider the vertices of $T$ level by level, starting from nodes in level $V_{\ell=h-1}$
and going up, modifying the solution $\OPT_{\ell+1}$ to obtain $\OPT_{\ell}$. 

\begin{definition}
For a node $v$, the $i$-th \text{bucket}, $b_i$, contains the number of tours of $\OPT_\ell$ having coverage between $[\sigma_i, \sigma_{i+1})$ tokens in $T_v$ where $\sigma_i$ is the $i$-th threshold value. We will denote a node and bucket by a pair $(v,b_i)$. Let $n_{v,i}$ be the number of tours in bucket $b_i$ of $v$. 
\end{definition}

\begin{definition}
A bucket $b$ is \textbf{small} if the number of tours in $b$ is at most $\alpha \log^3 n / \eps^2$ and is \textbf{big} otherwise,
for a constant $\alpha \ge \max\{1, 12\delta\}$.
\end{definition}

Note that for every node $v$ and bucket $b_i$ and for any two partial tours in $b_i$, the ratio of their size (coverage)
is at most $(1 + \eps)$. We will use this fact crucially later on.
While giving the high level idea earlier in this section, we mentioned that we can cover the orphant 
tokens at low cost by using a few extra tours at low cost. For this to work, we need to assume that the ratio of the maximum size tour to the minimum size tour 
in all groups $G^v_1,\ldots,G^v_g$ is at most $(1 + \eps)$. To have this property, 
we need to do the grouping described for each vertex-bucket pair $(v,b_i)$ that is big. 

For each $v \in V_\ell$, let $(v,b_i)$ be a vertex-bucket pair. If $b_i$ is a small bucket, we do not modify the partial tours in it. 
If $b_i$ is a big bucket, we create groups $G^v_{i,1}, \ldots, G^v_{i,g}$ of equal sizes 
(by adding null/empty tours if needed to $G^v_{i,1}$ to have equal size groups), for $g = (2\delta \log n)/\eps^2 $; so 
|$G^v_{i,j}|=\lceil n_{v,i}/g\rceil$. We also consider a mapping $f$ (as before) which maps (in the same order) the 
tours $t\in G^v_{i,j}$ to the tours in $G^v_{i,j-1}$ for all $1<j\leq g$.  We assume the mapping maps tours of $G^v_{i,1}$ to 
empty tours. 
Let the size of the smallest (largest) partial tour in $G^v_{i,j}$ be $h^{v,min}_{i,j}$ ($h^{v,max}_{i,j}$). Note that
$h^{v,max}_{i,j-1}\leq h^{v,min}_{i,j}$.
Consider the set ${\bf T}_\ell$ of all the tours $\calT$ in $\OPT_\ell$ that visit a vertex in one of the lower levels $V_{\geq \ell}$. Consider an arbitrary such tour $\calT$ that has a partial tour $t$ in a big vertex/bucket pair $(v,b_i)$, 
suppose $t$ belongs to group $G^v_{i,j}$. We replace $t$ with $f(t)$ in $\calT$.
Note that for $\calT$, the partial tour at $T_v$ now has a size between 
$h^{v,min}_{i,j-1}$ and $h^{v,max}_{i,j-1}$. Now, add some extra tokens at $v$ to be picked up by $\calT$ so that
the size of the partial tour of $\calT$ at $T_v$ is exactly $h^{v,max}_{i,j-1}$; note that
 since $h^{v,max}_{i,j-1}\leq |t|$, the new partial tour at $v$ can pick up the extra tokens without violating the capacity 
of $\calT$. 
If we make this change for all tours $\calT\in {\bf T}_\ell$, each partial tour of them at level $\ell$ that was in a group
$j<g$ of a big vertex/bucket pair $(v,i)$ is replaced with a smaller partial tour from group $j-1$ of the same big vertex/bucket
pair; after adding extra tokens at $v$ (if needed) the size is the maximum size from group $j-1$. 
All other partial tours (from small vertex/bucket pairs) remain unchanged. Also, the total cost of the tours has not increased (in fact some now have partial tours that are empty). However, the tokens that were picked by
partial tours from $G^v_{i,g}$ for a big vertex/bucket pair $(v,b_i)$ are now orphant. We describe how to cover them with some new tours. 

One important observation is that when we make these changes, for any partial tours at vertices at lower levels ($V_{>\ell}$)
their size remains the same. It is only the tour sizes going down a vertex at level $\ell$ that we are adjusting (by adding extra
tokens). All other lower level partial tours remain unchanged (only their top parts may get swapped). This property holds inductively
as we go up the tree and ensure that the lower level partial tours have one of polylogarithmic many sizes.
More precisely, as we go up levels to compute  $\OPT_\ell$, for any vertex $v'\in V_{\ell'}$ (where $\ell'>\ell$) and any partial
tour $\calT'$ visiting $T_{v'}$, either $|\calT'|$ belongs to a small vertex bucket pair $(v',b_{i'})$ (and so has one of
$O(\log^3 n/\eps)$ many possible values) or if it belongs to a big
vertex bucket pair $(v',b_{i'})$ then its size is equal to $h^{v',max}_{i',j'}$ for some group $j'$ and hence one of
$O((\log Q\log n)/\eps^2)$ possible values.

To handle (cover) orphant nodes, we are going to (randomly) select a subset of tours of $\OPT$ as "extra tours" and add them to $\OPT'$ and modify them
such that they cover all the tokens that are now orphant (i.e. those that were covered by partial tours of $G^v_{i,g}$ for all big vertex/bucket pairs at level $\ell$).

Suppose we select each tour $\calT$ of $\OPT$ with probability $\eps$. We make two copies of the extra tour and we designate both extra copies to one of the levels $V_\ell$ that it visits with equal probability. We call these the extra tours. 

\begin{lemma}\label{lem:extra-cost}
The cost of extra tours selected is at most $4\eps\cdot\opt$ w.h.p.
\end{lemma}
\begin{proof}
Recall that $f(e)$ denotes the number of tours passing through $e$ in $\OPT$. The contribution of edge $e$ to the optimal solution is $2\cdot w(e)\cdot f(e)$ and we can write $\opt = \sum_{e \in E} 2 \cdot w(e) \cdot f(e)$. Let $e$ be the parent edge of a node in $v \in V_\ell$. 
Suppose an extra tour is designated to level $\ell$, we will only use it to cover orphant tokens from big buckets from nodes in $V_\ell$. A node $v$ would use an extra tour to cover orphant tokens only if one of $v$'s buckets is a big bucket. From now on, we will assume the extra tours only pass through an edge $e$ if $f(e) \ge \alpha \log^3 n/\eps^2$ (we can shortcut it otherwise).

For an edge $e$, let $f'(e)$ denote the number of sampled tours passing through $e$ and since we use two copies of each sampled tour,  $2f'(e)$ is the number of extra tours passing through $e$ in $\OPT'$. We can write $\opt' = \sum_{e \in E} 2 \cdot w(e) \cdot (f(e) + 2f'(e))$ and the cost of extra tours is $\sum_{e \in E} 2 \cdot w(e) \cdot 2f'(e)$. While modifying $\OPT$ to $\OPT'$, each tour in the optimal solution is sampled with probability $\eps$. Let $e$ be an edge with $f(e)$ tours $\calT_{e,1}, \ldots, \calT_{e,f(e)}$ passing through it. Let $Y_{e,i}$ be a random variable which is 1 if tour $\calT_{e,i}$ is sampled and $0$ otherwise. 
$$\Ex{Y_{e,i}} = \Prob{\calT_{e,i} \text{ is sampled}} = \eps.$$
Let $f'(e) = Y_e = \sum_{i = 1}^{f(e)} Y_{e,i}$. By linearity of expectations, we have 
$$\Ex{f'(e)} = \Ex{Y_e} = \sum_{i = 1}^{f(e)} \Ex{Y_{e,i}} = \sum_{i =1}^{f(e)} \eps = \eps \cdot f(e).$$
Our goal is to show $\Prob{Y_e > 2\Ex{Y_e}}$ is very low. Using Chernoff bound with $\mu =\Ex{Y_e}=  \eps \cdot f(e) \ge \alpha \log^3 n /\eps \ge 6 \log n$.
\begin{equation*}
    \begin{split}
        \Prob{Y_e > 2\Ex{Y_e}} & \le e^{-(2\log n)} = \frac{1}{n^2} 
    \end{split}
\end{equation*}
The above concentration bound holds for a single edge $e$. Using the union bound, we can show this hold with high probability over all edges, 
$$\sum_{e \in E} \Prob{Y_e > 2\Ex{Y_e}} \le \frac{1}{n}.$$
We showed $f'(e) \le 2\eps \cdot f(e)$ with high probability. Hence, with high probability, the cost of the extra tours is at most 
$$\sum_{e \in E} 2 \cdot w(e) \cdot 2f'(e) \le \sum_{e \in E} 2 \cdot w(e) \cdot 4\eps \cdot f(e) = 4\eps \sum_{e \in E}2 \cdot w(e) \cdot f(e) = 4\eps \cdot \opt.$$
\end{proof}

Therefore, we can assume that the cost of all the extra tours added  is at most $4\eps\cdot\opt$.
Let $X_\ell$ be the set of extra tours designated to level $\ell$. 
We assume we add $X_\ell$ when we are building $\OPT_\ell$
(it is only for the sake of analysis). For each $v\in V_\ell$ and vertex/bucket pair $(v,b_i)$, let $X_{v,i}$ be those in $X_\ell$ whose partial tour in $T_v$ has
a size in bucket $b_i$. Each extra tour in $X_\ell$
will not be picking any of the tokens in levels $V_{<\ell}$ (as they will be covered by the tours already in $\OPT_\ell$); 
they are used to cover the orphant tokens created by partial tours
of $G^v_{i,g}$ for each big vertex/bucket pair $(v,b_i)$ with $v\in V_\ell$; as described below. 


\begin{lemma}\label{lem:extra}
For each level $V_\ell$, each vertex $v\in V_\ell$ and big vertex/bucket pair $(v,b_i)$,
w.h.p. $|X_{v,i}|\geq \frac{\eps^2}{\delta\log^2 n}\cdot n_{v,i}$.
\end{lemma}

\begin{proof}
Suppose $(v,b_i)$ is a big vertex/bucket pair at some level $V_\ell$. Let $p_1, \ldots, p_{n_{v,i}}$ be the partial tours in vertex/bucket pair $(v,b_i)$. Let the tour in $\OPT$ corresponding to $p_i$ be $\calT$. Two copies of tour $p_i$ are assigned to $b_i$ if both of the following events are true: 
\begin{itemize}
    \item Let $A_i$ be the event where tour $\calT$ is sampled as an extra tour. Since each tour is sampled with probability $\eps$, we have $\Prob{A_i} = \eps$. 
    \item Let $B_i$ be the event where tour $\calT$ is assigned to level $\ell$. There are $h = \delta \log^2 n/\eps $ many levels and since $\calT$ (if sampled) is assigned to any one of its levels, $\Prob{B_i} \ge 1/h \ge \eps/(\delta \log^2 n)$. 
\end{itemize}
Let $Y_i$ be a random variable which is 1 if $p_i$ is an extra tour in $(v,b_i)$ and 0 otherwise. 
$$\Ex{Y_i} = \Prob{Y_i = 1} = \Prob{A_i \land B_i}  = \Prob{A_i} \cdot \Prob{B_i} \ge \eps^2/ (\delta \log^2 n).$$
Let $Y_{v,i} = \sum_{i = 1}^{n_{v,i}} Y_i$ be the random variable keeping track of the number of sampled tours in $(v,b_i)$. The number of extra tours, $|X_{v,i}| = 2Y_{v,i}$ since we add two copies of a sampled tour to $X_{v,i}$. By linearity of expectation, we have 
$$\Ex{|X_{v,i}|} = 2\Ex{Y_{v,i}} = 2\sum_{i = 1}^{n_{v,i}} \Ex{Y_i} \ge \frac{2\eps^2}{\delta \log^2 n} \cdot n_{v,i}.$$
We want to show that $|X_{v,i}| \ge \frac{\Ex{|X_{v,i}|}}{2} \ge \frac{\eps^2}{\delta \log^2 n} \cdot n_{v,i} $ with high probability over all vertex-bucket pairs. 

Using  Chernoff Bound with $\mu = \Ex{|X_{v,i}|} \ge \frac{2\eps^2}{\delta \log^2 n} \cdot n_{v,i} \ge 24 \log n$ since $n_{v,i} \ge \alpha \log^3n /\eps^2$ and $\alpha \ge 12\delta$. 
\begin{equation*}
    \begin{split}
        \Prob{|X_{v,i}| < \frac{\Ex{|X_{v,i}|}}{2}} & \le e^{-(3\log n)} = \frac{1}{n^3}
    \end{split}
\end{equation*}
Note that the above equation only shows the concentration bound for a single vertex/bucket pair. There are $n$ nodes and each node has up to $\tau = \log n /\eps$ buckets, so the total number of vertex/bucket pairs is at most $n \log n/\eps$. Suppose we do a union bound over all buckets, we get
$$\sum_{\text{all } (v,b_i) \text{ pairs}} \Prob{|X_{v,i}| < \frac{\Ex{|X_{v,i}|}}{2}} \le \frac{1}{n}.$$
We showed that for each vertex/bucket pair $v,b_i$, $|X_{v,i}|\geq \frac{\eps^2}{\delta\log^2 n}n_{v,i}\geq \alpha\log n/(2\delta)$ holds with high probability. 
\end{proof}

\begin{lemma}\label{lem:extra2}
Consider all $v\in V_\ell$, big vertex/bucket pairs $(v,b_i)$ and partial tours in $G^v_{i,g}$.
We can modify the tours in $X_{v,i}$ (without increasing the cost) and adding some extra tokens at $v$ (if needed) so that:
\begin{enumerate}
\item The tokens picked up by partial tours in $G^v_{i,g}$ are covered by some tour in $X_{v,i}$, and
\item The new partial tours that pick up the orphant tokens in $G^v_{i,g}$ 
have size exactly $h^{v,max}_{i,g}$ and all tours still have size at most $Q$.
\item For each (new) partial tour of $X_{v,i}$ and every level $\ell'>\ell$, the size of partial tours of $X_{v,i}$ at a
vertex at level $\ell'$ is also one of $O(\log Q\log^3 n/\eps^3)$ many sizes.
\end{enumerate}
\end{lemma}
\begin{proof}
Our goal is to use the extra tours in $X_{v,i}$ to cover tokens picked up by partial tours of $G^v_{i,g}$ and we want each extra tour in $X_{v,i}$ to cover exactly $h^{v,max}_{i,g}$ tokens. The tours in the last group, $G^v_{i,g}$, cover $\sum_{t \in G^v_{i,g}} |t|$ many tokens. Since we want each tour in $X_{v,i}$ to cover $h^{v,max}_{i,g}$ tokens, we will add $\sum_{t \in G^v_{i,g}} (h^{v,max}_{i,g} - |t|)$ extra tokens at $v$ for each vertex/bucket pair $(v,b_i)$ so that there are $h^{v,max}_{i,g}$ tokens for each partial tour in $G^v_{i,g}$.  From now on, we will assume each partial tour in a last group $G^v_{i,g}$ covers $h^{v,max}_{i,g}$ tokens.

We know $|G^v_{i,g}| = n_{v,i}/g = \frac{\eps^2 }{2\delta \log n} \cdot n_{v,i}$. Using Lemma \ref{lem:extra}, we know with high probability that $|X_{v,i}| \ge \frac{\eps^2 }{\delta\log^2 n} \cdot n_{v,i} =2 |G^v_{i,g}|$, so $|X_{v,i}|/|G^v_{i,g}| \ge 2$. Recall $\OPT'$ includes tours in $\OPT$ plus the extra tours in $\OPT$ that were sampled.  Let $Y_{v,i}$ denote the number of tours in vertex/bucket pair $(v,b_i)$ that were sampled, so $|X_{v,i}| = 2|Y_{v,i}|$ since we made two extra copies of each sampled tour and $|Y_{v,i}|\ge |G^v_{i,g}|$ with high probability. We will start by creating a one-to-one mapping $s : G^v_{i,g} \rightarrow Y_{v,i}$ which maps each tour in $G^v_{i,g}$ to a sampled tour in $Y_{v,i}$. We know such a one-to-one mapping exists since $|Y_{v,i}|\ge |G^v_{i,g}|$. 

Let $\calT$ be a sampled tour in $Y_{v,i}$ with two extra copies of it, $\calT_1$ and $\calT_2$ in $X_{v,i}$. Let the partial tours of $\calT$ at the bottom part in $V_\ell$ be $p_1, \ldots, p_m$. We know $|\calT| \ge \sum_{i = 1}^m |p_i|$. Since $s$ is one-to-one, one partial tour from $r_k \in G^v_{i,g}$ maps to $p_j$ or no tour maps to $p_j$. If no tour maps to $p_j$, we consider the load assigned to $p_j$ to be zero. If $s(r_k) = p_j$ where  $r_k \in G^v_{i,g}$, since we added extra tokens to make each partial tour $r_k \in G^v_{i,g}$ have $h^{v,max}_{i,g}$ tokens, the load assigned to $p_j$ would be $h^{v,max}_{i,g}$. 

Suppose we think of $r_1, \ldots, r_m$ as items and $\calT_1$ and $\calT_2$ as bins of size $Q$. We know each $r_i$ fits into a bin of size $Q$. Recall that for the tour $r_j$ assigned to $p_j$, we know $|r_j| \le (1 + \eps)|p_j|$ since both $r_j$ and $p_j$ are in the same group  $G^v_{i,g}$. We might not be able to fit all items $r_1, \ldots, r_m$ into a bin of size $Q$ because $\sum_{i = 1}^m|r_i| \le (1 + \eps)\sum_{i = 1}^m |p_i| \le (1 + \eps)|\calT| \le (1 + \eps)Q$. However, if we used two bins of size $Q$, we can pack the items into both bins without exceeding the capacity of either bin such that each item $r_i$ is completely in one bin. Since $\calT_1$ and $\calT_2$ are not assigned to any lower level, they have not been used to cover any tokens so far in our algorithm and they both have unused capacity $Q$. Using the bin packing analogy, we could split $r_1, \ldots, r_m$ between $\calT_1$ and $\calT_2$. We could assign $r_1, \ldots, r_j$ (for the maximum $j$) to $\calT_1$ such that $\sum_{i = 1}^j |r_i| \le Q$ and the rest, $r_{j+1}, \ldots, r_m$ to $\calT_2$. Since $\sum_{i = 1}^m |r_i| \le (1 + \eps)Q$, we can ensure we can distribute the tokens in $r_i$'s amongst $\calT_1$ and $\calT_2$ such that both $\calT_1$ and $\calT_2$ cover at most $Q$ tokens. Although there are two copies of each partial tour $p_i$ in $X_{v,i}$, according to our approach, we are using at most one of them (their coverage would be zero if they are not used). If the coverage of one of the extra partial tours is non-zero, we also showed that if it picks up tokens from a partial tour in $G^v_{i,g}$, it would pick up exactly $h^{v,\max}_{i,g}$ tokens, proving the 2nd property of the  Lemma. 

Also, note that for each partial tour $r_k\in G^v_{i,g}$ and for each level $\ell'>\ell$ if $r_k$ visits a vertex $v'\in V_{\ell'}$,
then the partial tour of $r_k$ at $T_{v'}$ already satisfies the properties that: either its size belongs to a small vertex-bucket
pair $(v',b_i)$ (so has one of $O(\log^3 n/\eps)$ many possible values) or if it belongs to a big
vertex bucket pair $(v',b_{i'})$ then its size is equal to $h^{v',max}_{i',j'}$ for some group $j'$ and hence one of
$O((\log Q\log n)/\eps^2)$ possible values.
This implies that for the extra tours of $X_{v,i}$, after we reassign partial
tours of $G^v_{i,g}$ to them (to cover the orphant nodes), each will have a size exactly equal to $h^{v,max}_{i,g}$ at level $\ell$
and at lower levels $V_{>\ell}$ they already have one of the $O(\log Q\log^3 n/\eps^3)$ many possible sizes. This establishes the 3rd property of the lemma.
\end{proof} 

Therefore, using Lemma \ref{lem:extra2}, all the tokens of $T_v$ remain covered by partial tours;
those partial tours in $G^v_{i,j}$ (for $1\leq j<g$) are tied to the top parts of the tours from group $G^v_{i,j+1}$ and the
partial tours of $G^v_{i,g}$ will be tied to extra tours designated to level $\ell$. We also add extra tokens at $v$ to be
picked up by the partial tours of $T_v$ so that each partial tour has a size exactly equal to the maximum size of a group.
All in all, the extra cost paid to build $\OPT_\ell$ (from $\OPT_{\ell+1}$) is for the extra tours designated to level $\ell$.

\begin{theorem}\label{lem:struct1} \textbf{(Structure Theorem)}
Let $\opt$ be the cost of the optimal solution to instance $\calI$. We can build an instance $\calI'$ on the same tree $T$
such that each node has $\geq 1$ tokens and there exists a near-optimal solution $\OPT'$ for $\calI'$
having cost $(1 + 4\eps)\opt$ w.h.p with the following property.
The partial tours going down subtree $T_v$ for every node $v$ in $\OPT'$ has one of $O((\log Q \log^3 n)/\eps^3)$ possible sizes.
More specifically, suppose $(v,b_i)$ is a bucket pair for $\OPT'$. Then either:
\begin{itemize}
\item $b_i$ is a small bucket and hence there are at most $\alpha \log^3 n/\eps^2$ many partial tours of $T_v$ whose size is in bucket $b_i$, or
\item $b_i$ is a big bucket; in this case there are $g = (2\delta \log n)/\eps^2 $ many group sizes in $b_i$: 
$\sigma_i \leq h^{v,max}_{i,1}\leq \ldots\leq h^{v,max}_{i,g}<\sigma_{i+1}$ and every tour of bucket $i$ has one of these sizes. 
\end{itemize}
\end{theorem}

\begin{proof}
We will show how to modify $\OPT$ to a near-optimal solution $\OPT'$. We start from $\ell = h$ and let $ \OPT_\ell = \OPT$. For decreasing values of $\ell$ we show, for each $\ell$ how to modify $\OPT_{l+1}$ to obtain $\OPT_\ell$. We do this in the following manner: we do not modify partial tours in small buckets. However, for tours in big buckets, in each vertex/bucket pair $(v,b_i)$ in level $\ell-1$, we place them into $g$ groups $G_1^v, \ldots, G^v_g$ of equal sizes by placing the $i$'th $n_v/g$ partial tours into $G_i^v$. We have a mapping $f$ from each partial tour in $G^v_{i-1}$ to one in $G^v_{i}$ for $i \in \{2, \ldots, g\}$. We modify $\OPT_\ell$ to $\OPT_{\ell+1}$ in the following way: for each tour $\calT$ that has a partial tour $t \in G_i^v$, replace the bottom part of $\calT$ at $v$ from $t$ to $f(t)$ (which is in $G^v_{i-1}$). For each tour $t \in G^v_{i-1}$, we will add $h^{v, \max}_{i-1} - |t|$ many extra tokens at $v$. Note that by this change, the size of any tour such as $\calT$ can only decrease and we are not violating feasibility of the tour because $h^{v, \max}_{i-1} \le h_i^{v, \min}$. However, the tokens in $T_v$ picked up by the partial tours in $G^v_{i,g}$ are not covered by any tours. We can use Lemma \ref{lem:extra2} to show how we can use extra tours to cover the partial tours in $G^v_{i,g}$ such that the new partial tours have size exactly $h^{v,\max}_{i,g}$.  

We will inductively repeat this for levels $\ell-2, \ell-3, \ldots, 1$ and obtain $\OPT_1 = \OPT'$. Note that by adding extra tokens $h^{v, \max}_{i-1} - |t|$ for a tour $t \in  G^v_{i-1}$, we are enforcing that the coverage of each tour is the maximum size of tours in its group. In a big bucket, there are $g = (2\delta \log n)/\eps^2$ many group sizes, so there are $O(\log n/\eps^2)$ possible sizes for tours in big buckets at a node. In a small bucket, there can be at most $\alpha \log^3 n/\eps^2$ many tours and since there are $\tau = O(\log Q /\eps)$ many buckets, there can be at most $O((\log Q \log^3 n)/\eps^3)$ many tour sizes covering $T_v$.  

Using Lemma \ref{lem:extra-cost}, we know the cost of the extra tours is at most $4\eps \cdot \opt$ with high probability, so the cost of $\opt' 
\le (1 + 4\eps) \opt$.
\end{proof}


\subsection{Dynamic Program}

In this section we complete the proof of Theorem \ref{thm:tree}. We will describe how we can compute a solution of cost at most $(1 + 4\eps)\opt$ using dynamic programming and based on the existence of a near-optimum solution guaranteed using the structure theorem.
For each vertex/bucket pair, we do not know if the bucket is small or big, so we will consider subproblems corresponding to both possibilities. Informally, we will have a vector $\nvec \in [n]^{\tau}$ where if $i < 1/\eps$, $n_i$ keeps track of the exact number of tours of size $i$ and for $i \ge 1/\eps$, $\nvec_i$ keeps track of the number of tours in bucket $b_i$, or tours covering between $[\sigma_i, \sigma_{i+1})$ tokens. Let $o_v$ denote the total number of tokens to be picked up across all nodes in the subtree $T_v$. Since each node has at least one token, $o_v \ge |V(T_v)|$. We will keep track of three other pieces of information conditioned on whether $b_i$ is a small or big bucket. If $b_i$ is a small bucket, we will store all the tour sizes exactly. Since the number of tours in a small bucket is at most $\gamma = \alpha \log^3n /\eps^2$, we will use a vector $\tvec^i \in [n]^\gamma$ to represent the tours of a small bucket where $\tvec^i_j$ represents the size of $j$-th tour in bucket $b_i$. Suppose $b_i$ is a big bucket, there are $g = (2\delta \log n)/\eps^2 $ many tour sizes in the bucket corresponding to $n^g$ possibilities. For each big bucket $b_i$ at node $v$, we need to keep track of the following information, 
\begin{itemize}
    \item $\hvec^i_v \in [n]^g$ is a vector where $\hvec^i_{v,j} = h^{v, \max}_{i,j}$, which is the size of the maximum tour in group $j$ of bucket $i$ at node $v$. 
    \item $\lvec^i_v \in [n]^g$ is a vector where $\lvec^i_{v,j}$ denotes the number of partial tours covering $h^{v, \max}_{i,j}$ tokens which lies in group $j$ of bucket $i$ at node $v$. 
\end{itemize}
Let $\yvec_v$ denote a configuration of tours across all buckets of $v$. 
$$\yvec_v = [o_v, \nvec_v, (\tvec_v^1, \hvec_v^1, \lvec_v^1), (\tvec_v^2, \hvec_v^2, \lvec_v^2), \ldots, (\tvec_v^\tau, \hvec_v^\tau, \lvec_v^\tau)].$$
Note that a bucket $b_i$ is either small or big and cannot be both, hence given $(\tvec^i_v, \hvec^i_v, \lvec^i_v)$, it cannot be the case that $\tvec^i_v \neq \vec{0}, \hvec^i_v \neq \vec{0}$ and $\lvec^i_v \neq \vec{0}$. 
The subproblem $\abold[v, \yvec]$ is supposed to be the minimum cost collection of partial tours going down $T_v$ (to cover the tokens in $T_v$) and the cost of using the parent edge of $v$ having tour profile corresponding to $\yvec$. Our dynamic program heavily relies on the properties of the near-optimal solution in the structure theorem.  Let $v$ be a node. We will compute $A[\cdot, \cdot]$ in a bottom-up manner, computing $\abold[v, \yvec_v]$ after we have computed the entries for the children of $v$. 

The final answer is obtained by looking at the various entries of $\abold[r, \cdot ]$ and taking the smallest one. First, we argue why this will correspond to a solution of cost no more than $\opt'$. We will compute our solution in a bottom-up manner. 

For the base case, we consider leaf nodes. A leaf node $v$ with parent edge $e$ could have $o_v \ge 1$ tokens at $v$. We will set $\abold[v,\yvec_v] = 2 \cdot w(e) \cdot m_v$ where $m_v$ is the number of tours in $\yvec_v$ if the total sum of tokens picked up by the partial tours in $\yvec_v$ is exactly $o_v$. Recall that $f(e)$ is the load on (i.e. number of tours using) edge $e$.
From our structure theorem, we know there exists a near optimum solution such that each partial tour 
of $T_v$ has one of $O((\log Q \log^3 n)/\eps^3)$ tour sizes and for each small bucket, there are at most $\alpha \log^3n /\eps^2$ partial tours in it. For every big bucket, there are $g = (2 \delta \log n)/\eps^2$ many group sizes and every tour of bucket $i$ has one of these sizes. The base case follows directly from the structure theorem. 

To compute cell $\abold[v, \yvec_v]$, we would need to use another auxiliary table $\bbold$. Suppose $v$ has $k$ children $u_1, \ldots, u_k$ and assume we have already calculated $\abold[u_j, \yvec]$ for every $1 \le j \le k$ and for all vectors $\yvec$. Then we define a cell in our auxiliary table $\bbold[v,\yvec_v', j]$ for each $1 \le j \le k$ where $\bbold[v,\yvec_v', j]$ is the minimum cost of covering $T_{u_1} \cup \ldots \cup T_{u_j}$ where $\yvec_v'$ is the tour profile for the union of subtrees $T_{u_1} \cup \ldots \cup T_{u_j}$. In other words, $\bbold[v, \yvec_v', j]$ is what $\abold[v, \yvec_v]$ is supposed to capture when restricted only to the first $j$ children of $v$. We will set $\abold[v, \yvec_v] = \bbold[v, \yvec_v', k] + 2 \cdot w(e) \cdot m_v$ where $m_v$ is the number of different tours in $\yvec_v'$. We will assume the parent edge of the depot has weight 0. Suppose $T_{u_i}$ has $o_i$ tokens, then the number of tokens in $T_v$ is at least $1 + \sum_{i = 1}^k o_i$. To compute entries of $\bbold[v, \cdot, \cdot]$, we use both $\abold$ and $\bbold$ entries for smaller subproblems of $v$ in the following way:  

\textbf{Case 1:} j = 1: This is the case when we restrict the coverage to only the first child of $v$, $u_1$. 
\begin{equation*}
    \begin{split}
        \bbold[v,\yvec_v',1] & = \mi{\yvec'} \left\{ \abold[u_1, \yvec'] \right\} \\
    \end{split}
\end{equation*}
We will find the minimum cost configurations $\yvec'$ such that $\yvec_v'$ and $\yvec'$ are consistent with each other. We say $\yvec_v'$ and $\yvec'$ are consistent if a tour in $\yvec_v'$ either only covers tokens at $v$ and does not visit any node below $v$ or $\yvec_v'$ consists of a tour from $\yvec'$ plus zero or more extra tokens picked up at $v$. Moreover, every tour in $\yvec'$ is part of some tour in $\yvec_v'$. 

\textbf{Case 2:} $2 \le j \le k$. We will assume we have computed $\bbold[v, \yvec', j - 1]$ and $\abold[u_{j}, \yvec'']$ and we have 
$$\bbold[v, \yvec_v', j] = \mi{\yvec', \yvec''} \{\bbold[v, \yvec', j - 1] + \abold[u_{j}, \yvec''] \}.$$
There are four possibilities for each partial tour $t_v$ at node $v$ going down $T_v$ covering tokens for subtrees rooted at children $u_1, \ldots, u_k$ . 
\begin{itemize}
    \item $t_v$ could be a tour that only picks up tokens at $v$ and does not pick up tokens from subtrees $T_{u_1} \cup \ldots \cup T_{u_{j}}$.
    \item $t_v$ could be a tour that picks up tokens at $v$ and picks up tokens only from subtrees $T_{u_1} \cup \ldots \cup T_{u_{j-1}}$. 
    \item $t_v$ could be a tour that picks up tokens at $v$ and picks up tokens only from subtree $T_{u_{j}}$.
    \item $t_v$ could be a tour that picks up tokens at $v$ and picks up tokens from subtrees $T_{u_1} \cup \ldots \cup T_{u_{j}}$. 
\end{itemize}
We would find the minimum cost over all configurations $\yvec_v', \yvec'$ and $\yvec''$ as long as $\yvec_v', \yvec'$ and $\yvec''$ are consistent. We say tours $\yvec_v',\yvec'$ and $\yvec''$ are consistent if there is a way to combine partial tours from $\yvec'$ and $\yvec''$ to form a partial tour in $\yvec_v'$ while also picking up extra tokens at node $v$. We will define consistency more rigorously in the next section. 

\subsection{Checking Consistency}
In our dynamic program, for the inner DP, we are given three vector $\yvec_v', \yvec', \yvec''$ where $v$ is a node having children $u_1, \ldots, u_{j}$. $\yvec'$ represents the configuration of tours in $T_{u_1} \cup \ldots \cup T_{j-1}$ and $\yvec''$ represents the configuration of tours covering $T_{u_{j}}$. For the case of checking consistency for case 1, we will assume $\yvec'' = \vec{0}$. Suppose we are given $o_v$ (for node $v$), $o_u$ for children $u_1,\ldots,u_{j-1}$, and
$o_w$ for $u_j$, we can infer that there are $o'_v = o_v - o_u - o_w$ extra tokens that need to be picked at $v$. $o'_v$ tokens need to be distributed amongst tours in $\yvec_v$. 
There are four possibilities for each tour $t_v$ in $\yvec_v'$. 
\begin{itemize}
    \item $t_v$ could be a tour that picks up extra tokens at $v$ and picks up tokens only from subtrees $T_{u_1} \cup \ldots \cup T_{u_{j-1}}$. 
    \item $t_v$ could be a tour that picks up extra tokens at $v$ and picks up tokens only from subtree $T_{u_{j}}$.
    \item $t_v$ could be a tour that picks up extra tokens at $v$ and picks up tokens from subtrees $T_{u_1} \cup \ldots \cup T_{u_{j}}$. 
\end{itemize}
For simplicity, we will refer to a tour picking up tokens in $T_{u_1} \cup \ldots \cup T_{u_{j-1}}$ to be $t_u$ and a tour picking up tokens from $T_{u_{j}}$ to be $t_w$.  
\begin{definition}
We say configurations $\yvec_v', \yvec'$ and $\yvec''$ are \textbf{consistent} if the following holds: 
\begin{itemize}
    \item Every tour in $\yvec'$ maps to some tour in $\yvec_v'$. 
    \item Every tour in $\yvec''$ maps to some tour in $\yvec_v'$. 
    \item Every tour in $\yvec_v'$ has at most two tours mapping to it and both tours cannot be from $\yvec'$ or $\yvec''$. 
    \item Suppose only one tour ($t_u$) maps to a tour $t_v$ in $\yvec_v'$. The number of extra tokens picked up by tour $t_v$ at $v$ is $|t_v| - |t_u|$. 
    \item Suppose $t_v$, a tour in $\yvec_v'$ has two tours: $t_u$ in $\yvec'$ and $t_w$ in $\yvec''$ mapped to it, then the number of extra tokens picked up by tour $t_v$ at $v$ is $|t_v| - |t_u| - |t_w|$. 
    \item The extra tokens at $v$, $o'_v = o_v - o_u - o_w$, are picked up by the tours in $\yvec_v'$. 
\end{itemize}
\end{definition}
Consistency ensures that we can patch up tours from subproblems and combine them into new tours in a correct manner while also picking up extra tokens at $v$. Now we will describe how we can compute consistency. Let $\zvec$ be a vector containing a subset of information contained in $\yvec$. 
$$\zvec_v = [\nvec_v, (\tvec_v^1, \hvec_v^1, \lvec_v^1), (\tvec_v^2, \hvec_v^2, \lvec_v^2), \ldots, (\tvec_v^\tau, \hvec_v^\tau, \lvec_v^\tau)].$$
From now on, we will choose to not write $\nvec_v$ explicitly since we can figure out the entries of the vector from $\lvec$. Suppose $|t_v|$ is the length of a tour in $\zvec_v'$. Let $\zvec_v' - t_v$ refer to the configuration $\zvec_v'$ having one less tour of size $|t_v|$. 
Let $\cbold[o'_v, \zvec_v', \zvec', \zvec''] = $ True if it is consistent and False otherwise. For the base case, $\cbold[0, \vec{0}, \vec{0}, \vec{0}] = $True. For the recurrence, we will look at all possible ways of combining $\zvec'$ and $\zvec''$ into $\zvec_v'$ while also picking up extra tokens $o'_v$. Note that $t_v$ is always non-zero, but both or one of $t_u$ or $t_w$ could be zero. 
$$\cbold[o'_v, \zvec_v', \zvec', \zvec''] = \underset{\substack{t_v, t_u, t_w \\ |t_v| = |t_u| + |t_w| + o_c}}{\bigvee} \cbold[o'_v - o_c, \zvec_v' - t_v, \zvec' - t_u, \zvec'' - t_w].$$

\subsection{Time Complexity}
We will work bottom-up and assume we have already pre-computed our consistency table. Computing $\bbold[\cdot, \cdot, \cdot]$ requires looking at previously computed $\bbold[\cdot, \cdot, \cdot]$ and  $\abold[\cdot,\cdot]$. Given $\yvec_v', \yvec'$ and $\yvec''$ which are all consistent, computing the cost of $\yvec_v'$ using $\yvec'$ and $\yvec''$ takes $O(1)$ time. Each $\yvec_v'$ consists of 
\begin{enumerate}
    \item $\nvec$ has $n^{O(\log n / \eps)}$ possibilities. 
    \item Each $\tvec^i$ has $n^{O(\log^3 n /\eps^2)}$ possibilities since there are $O(\log^3 n /\eps)$ tours in a small bucket. 
    \item Each $\hvec$ and $\lvec$ have $n^{O(g)}$ possibilities. Recall that $g = (2\delta \log n)/\eps^2 $, so each $\hvec$ and $\lvec$ have $n^{O(\log n /\eps^2)}$ possibilities. 
    \item Each triple $(\tvec^i, \hvec^i, \lvec^i)$ has $n^{O(\log^3 n /\eps^2)}$ possibilities. 
    \item $(\tvec^1, \hvec^1, \lvec^1), (\tvec^2, \hvec^2, \lvec^2), \ldots, (\tvec^\tau, \hvec^\tau, \lvec^\tau)$ have $n^{O(\tau \log^3 n /\eps^2)} = n^{O((\log Q\log^3 n)/\eps^3)}$ possibilities since $\tau = O(\log Q/\eps)$. 
\end{enumerate}
In total, each $\yvec_v'$ has $n^{O((\log Q\log^3 n)/\eps^3)}$ possibilities. For each $\yvec_v'$, we will have $n^{O((\log Q\log^3 n)/\eps^3)}$ possibilities for $\yvec_u$ and $\yvec_w$. Since there are $n^{O((\log Q\log^3 n)/\eps^3)}$ possibilities for $\yvec_v'$, the cost of computing the DP entries for a single node $v$ would be $n^{O((\log Q\log^3 n)/\eps^3)}$ and since there are $n$ nodes in the tree, the total time of computing the DP table assuming the consistency table is precomputed is $n^{O((\log Q\log^3 n)/\eps^3)}$. 

Before we compute our DP, we will first compute the consistency table $\cbold[\cdot, \cdot, \cdot, \cdot]$. Similar to our DP table, each entry of the consistency table has $n^{O((\log Q\log^3 n)/\eps^3)}$ possibilities. Assuming we have already precomputed smaller entries of $\cbold$ , there are $n^{O((\log Q\log^3 n)/\eps^3)}$ ways of picking $t_v, t_u$ and $t_w$. For a fixed $\yvec_v, \yvec_u, \yvec_w$ and $o'_v$, computing $\cbold[o'_v, \zvec_v', \zvec', \zvec'']$ takes $n^{O((\log Q\log^3 n)/\eps^3)}$ time. Since there are only $n^{O((\log Q\log^3 n)/\eps^3)}$ possibilities for $\zvec_v', \zvec'$ and $\zvec''$, the cost of computing all entries of the consistency table is $n^{O((\log Q\log^3 n)/\eps^3)}$. 

The time for computing both the DP table and consistency table is $n^{O((\log Q\log^3 n)/\eps^3)}$, so the total time taken by our algorithm is $n^{O((\log Q\log^3 n)/\eps^3)}$. For the unit demand case, since $Q \le n$, the runtime of our algorithm is $n^{O(\log^4 n/\eps^3)}$.

\subsection{Extension to Splittable CVRP}
We can extend our algorithm for unit demand CVRP in trees and show how we can get a QPTAS for splittable CVRP as long as the demands are quasi-polynomially bounded (Corollary \ref{cor:tree}). In our algorithm for unit demand CVRP, we viewed the demand of each node as a token placed at the node. For splittable CVRP, we could assume each node has $1 \le d(v) < nQ$ tokens and we can use the same structure theorem as before by modifying tours such that there are at most $O((\log Q \log^3 n)/\eps^3)$ different tour sizes for partial tours at a node. We can use the same DP to compute the solution. Each $\yvec_v$ consists of 
\begin{enumerate}
    \item $\nvec$ has $(nQ)^{O(\log n / \eps)}$ possibilities. 
    \item Each $\tvec^i$ has $(nQ)^{O(\log^3 n /\eps^2)}$ possibilities since there are $O(\log^3 n /\eps)$ tours in a small bucket. 
    \item Each $\hvec$ and $\lvec$ have $(nQ)^{O(g)}$ possibilities. Recall that $g = (2\delta \log n)/\eps^2 $, so each $\hvec$ and $\lvec$ have $(nQ)^{O(\log n /\eps^2)}$ possibilities. 
    \item Each triple $(\tvec^i, \hvec^i, \lvec^i)$ has $(nQ)^{O(\log^3 n /\eps^2)}$ possibilities. 
    \item $(\tvec^1, \hvec^1, \lvec^1), (\tvec^2, \hvec^2, \lvec^2), \ldots, (\tvec^\tau, \hvec^\tau, \lvec^\tau)$ have $(nQ)^{O(\tau \log^3 n /\eps^2)} = (nQ)^{O((\log Q\log^3 n)/\eps^3)}$ possibilities since $\tau = O(\log Q/\eps)$. 
\end{enumerate}
Similar to the analysis of the runtime of the unit demand case, the time complexity of computing the entries of DP tables $\abold, \bbold$, and the consistency table $\cbold$ is, $(nQ)^{O((\log Q\log^3 n)/\eps^3)}$. Suppose $Q = n^{O(\log^c n)}$, then the runtime of our algorithm is $n^{O(\log^{2c+4} n/\eps^3)}$.

\subsection{Height reduction}\label{sec:heightreduction}
In this section, we will prove Theorem \ref{thm:height-red}.
The first goal is to decompose the edge set of the tree $T$ into edge-disjoint paths. We will do so using the following lemma, similar to Lemma 5 from Cygan et al. \cite{DBLP:conf/esa/CyganGLPS12} to obtain such a decomposition in polynomial-time for a different problem.
\begin{lemma}\label{lem:height-de}
There exists a decomposition of the edge set of $T$ into edge-disjoint paths which can be grouped into $s = O(\log n)$ collections (called levels) $L_1, \ldots, L_s$ such that the following hold: 
\begin{enumerate}
    \item A root-to-leaf path $P$ in $T$ can be written as $P = Q_0Q_1 \ldots Q_s$ where $Q_i$ is either a path in $L_i$ or it is empty. 
    \item $P$ would use a path from a lower level $L_i$ before using a path from a higher level, $L_j$ where $i < j$. 
\end{enumerate}
\end{lemma}
\begin{proof}
Given a tree $T$, a D-path of $T$ is a root-to-leaf path $P= v_1v_2\ldots v_k$ such that $v_{i+1}$ is the child of $v_i$ with the largest number of nodes in the tree rooted at $T_{v_{i+1}}$. If there are multiple children with the same number of descendants, break ties arbitrarily. Let $P$ be a D-path. All the nodes in D-path $P$ receive label $1$. Let $T_1, \ldots, T_c$ be the set of trees obtained from $T - P$. Let $P_i$ be the D-path for $T_i$. We will label all nodes in $P_i$ to be 2. We will repeat this process recursively by finding D-paths for trees resulting from $T_i - P_i$ and labelling every node in the D-path with value corresponding to the depth of recursion. Each step involves finding a D-path, labelling the nodes of the path, deleting the path and recursively repeating the process for the resulting trees (with the value of the label increased by 1). Nodes of D-paths of trees at depth $\ell$ in the recursion receive labels $\ell$.  We will terminate this process when all nodes have been labelled. Let $L_j$ denote the collection of all D-paths whose nodes received the label $j$ (see figure \ref{fig:decomp}).

Note that after the first step, the trees $T_1, \ldots, T_c$ satisfy the property that $|V(T_i)| \le \frac{|V(T)|}{2}$ i.e., each tree is at most half of the original tree. This is because we pick the child with the largest number of nodes in the subtree rooted at it. After each step, the size of the new components formed is at most half the size of the previous component, hence we would use at most $\log n$ labels to label all nodes in the tree.
\end{proof} 

 The following is an example of such labelling where each color represents a level. 
\begin{figure}[H]
\centering
    \subfloat[\centering A tree before labelling.]{{\includegraphics[scale=0.42]{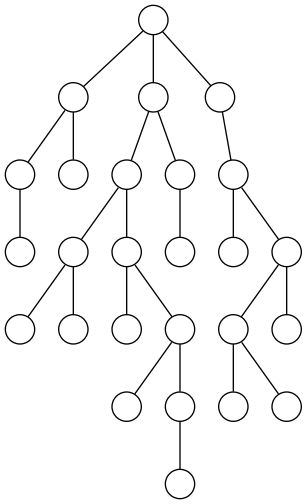} }}%
    \qquad
    \subfloat[\centering Blue edges are level 1, red edges are level 2 and green edges are level 3.]{{\includegraphics[scale=0.35]{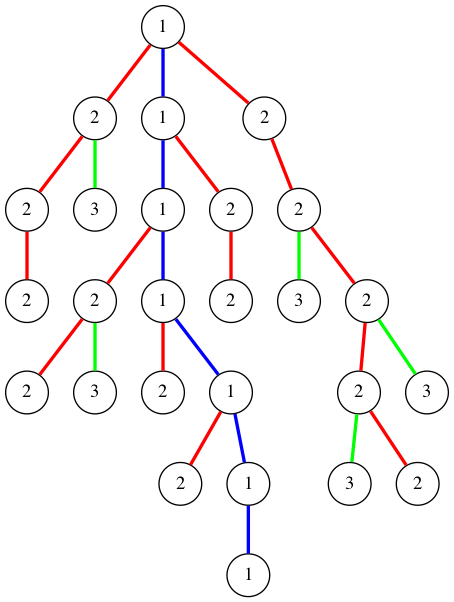} }}%
    \caption{An example of a tree before and after applying labels to nodes}%
    \label{fig:decomp}%
\centering
\end{figure}

\subsubsection{Creating a new tree}
Given a tree $T$, we can use Lemma \ref{lem:height-de} to decompose the tree into edge-disjoint paths. Next, we describe an algorithm to modify the tree recursively into a low height tree. The first step is to look at all the paths in $L_1$. $L_1$ is a special case since there is only one path in $L_1$ which goes from the depot to a leaf node. All the other levels $L_i$ could have multiple disjoint paths. Let $P$ be the path in $L_1$ and let $l(P)$ be the number of edges in path $P$. If $l(P) \le \frac{\delta\log n}{\eps}$ for
a $\delta>0$ to be specified, then we are done for $L_1$.

 However, if $l(P) > \frac{\delta\log n}{\eps}$, we will compress the path into a low height one. We will do a sequence of what is called up-pushes. We will pick 
 $s \leq \frac{\delta\log l(P)}{\eps}$ points to be \textbf{anchor points}. Let us call the anchor points $a_1, \ldots, a_s$ where $a_1$ is the anchor point closest to the root and $a_s$ is closest to the leaf. We will later show how to find these anchor points. 

 \begin{figure}[H]
\caption{A tree before an up-push (left) and after (right)  with reduced height. The blue edge connecting $a_i$ and $a_{i+1}$ has weight $w = w_p + w_q + w_s + w_t$}\label{fig:upush}
\footnotesize
\begin{tikzpicture}[scale=0.35,>=latex',line join=bevel,]
  \pgfsetlinewidth{1bp}
\pgfsetcolor{black}
  \draw [] (262.5bp,463.82bp) .. controls (262.5bp,452.76bp) and (262.5bp,438.4bp)  .. (262.5bp,427.08bp);
  \draw [] (255.35bp,389.29bp) .. controls (248.6bp,373.03bp) and (238.53bp,348.78bp)  .. (231.74bp,332.44bp);
  \definecolor{strokecol}{rgb}{0.0,0.0,0.0};
  \pgfsetstrokecolor{strokecol}
  \draw (249.5bp,362.5bp) node {$\ \ \ w_p$};
  \pgfsetcolor{red}
  \draw [] (269.65bp,389.29bp) .. controls (276.57bp,372.62bp) and (286.99bp,347.53bp)  .. (293.77bp,331.21bp);
  \definecolor{strokecol}{rgb}{0.0,0.0,0.0};
  \pgfsetstrokecolor{strokecol}
  \draw (286.5bp,362.5bp) node {$\ \ \ w_1$};
  \draw [] (212.62bp,298.16bp) .. controls (198.35bp,280.31bp) and (174.53bp,250.54bp)  .. (160.29bp,232.74bp);
  \draw (193.5bp,265.5bp) node {$\ \ \ w_q$};
  \pgfsetcolor{red}
  \draw [] (224.5bp,294.23bp) .. controls (224.5bp,278.66bp) and (224.5bp,256.48bp)  .. (224.5bp,240.06bp);
  \definecolor{strokecol}{rgb}{0.0,0.0,0.0};
  \pgfsetstrokecolor{strokecol}
  \draw (226.5bp,265.5bp) node {$\ \ \ w_2$};
  \pgfsetcolor{red}
  \draw [] (238.25bp,299.39bp) .. controls (257.12bp,280.71bp) and (290.73bp,247.46bp)  .. (308.59bp,229.79bp);
  \definecolor{strokecol}{rgb}{0.0,0.0,0.0};
  \pgfsetstrokecolor{strokecol}
  \draw (280.5bp,265.5bp) node {$\ \ w_3$};
  \draw [] (135.81bp,201.98bp) .. controls (119.55bp,184.02bp) and (91.525bp,153.06bp)  .. (75.24bp,135.08bp);
  \draw (113.5bp,168.5bp) node {$\ \ w_s$};
  \pgfsetcolor{red}
  \draw [] (148.5bp,197.23bp) .. controls (148.5bp,181.66bp) and (148.5bp,159.48bp)  .. (148.5bp,143.06bp);
  \definecolor{strokecol}{rgb}{0.0,0.0,0.0};
  \pgfsetstrokecolor{strokecol}
  \draw (150.5bp,168.5bp) node {$\ \ \ w_4$};
  \pgfsetcolor{red}
  \draw [] (163.4bp,203.58bp) .. controls (185.37bp,185.27bp) and (226.19bp,151.26bp)  .. (247.95bp,133.12bp);
  \definecolor{strokecol}{rgb}{0.0,0.0,0.0};
  \pgfsetstrokecolor{strokecol}
  \draw (214.5bp,168.5bp) node {$\ \ w_5$};
  \draw [] (54.811bp,102.02bp) .. controls (47.032bp,84.831bp) and (35.071bp,58.404bp)  .. (27.264bp,41.157bp);
  \draw (46.498bp,71.5bp) node {$\ \ w_t$};
  \pgfsetcolor{red}
  \draw [] (70.186bp,102.02bp) .. controls (78.013bp,84.724bp) and (90.075bp,58.076bp)  .. (97.877bp,40.837bp);
  \definecolor{strokecol}{rgb}{0.0,0.0,0.0};
  \pgfsetstrokecolor{strokecol}
  \draw (89.498bp,71.5bp) node {$\ \ w_6$};
\begin{scope}
  \definecolor{strokecol}{rgb}{0.0,0.0,0.0};
  \pgfsetstrokecolor{strokecol}
  \draw (262.5bp,482.0bp) ellipse (18.0bp and 18.0bp);
\end{scope}
\begin{scope}
  \definecolor{strokecol}{rgb}{0.0,0.0,0.0};
  \pgfsetstrokecolor{strokecol}
  \draw (262.5bp,407.5bp) ellipse (19.5bp and 19.5bp);
  \draw (262.5bp,407.5bp) node {$a_i$};
\end{scope}
\begin{scope}
  \definecolor{strokecol}{rgb}{0.0,0.0,0.0};
  \pgfsetstrokecolor{strokecol}
  \draw (224.5bp,314.0bp) ellipse (19.5bp and 19.5bp);
  \draw (224.5bp,314.0bp) node {$p$};
\end{scope}
\begin{scope}
  \definecolor{strokecol}{rgb}{1.0,0.0,0.0};
  \pgfsetstrokecolor{strokecol}
  \draw (300.5bp,337.0bp) -- (261.51bp,302.5bp) -- (339.49bp,302.5bp) -- cycle;
  \definecolor{strokecol}{rgb}{0.0,0.0,0.0};
  \pgfsetstrokecolor{strokecol}
  \draw (300.5bp,314.0bp) node {$T_1$};
\end{scope}
\begin{scope}
  \definecolor{strokecol}{rgb}{0.0,0.0,0.0};
  \pgfsetstrokecolor{strokecol}
  \draw (148.5bp,217.0bp) ellipse (19.5bp and 19.5bp);
  \draw (148.5bp,217.0bp) node {$s$};
\end{scope}
\begin{scope}
  \definecolor{strokecol}{rgb}{1.0,0.0,0.0};
  \pgfsetstrokecolor{strokecol}
  \draw (224.5bp,240.0bp) -- (185.51bp,205.5bp) -- (263.49bp,205.5bp) -- cycle;
  \definecolor{strokecol}{rgb}{0.0,0.0,0.0};
  \pgfsetstrokecolor{strokecol}
  \draw (224.5bp,217.0bp) node {$T_2$};
\end{scope}
\begin{scope}
  \definecolor{strokecol}{rgb}{1.0,0.0,0.0};
  \pgfsetstrokecolor{strokecol}
  \draw (320.5bp,240.0bp) -- (281.51bp,205.5bp) -- (359.49bp,205.5bp) -- cycle;
  \definecolor{strokecol}{rgb}{0.0,0.0,0.0};
  \pgfsetstrokecolor{strokecol}
  \draw (320.5bp,217.0bp) node {$T_3$};
\end{scope}
\begin{scope}
  \definecolor{strokecol}{rgb}{0.0,0.0,0.0};
  \pgfsetstrokecolor{strokecol}
  \draw (62.5bp,120.0bp) ellipse (19.5bp and 19.5bp);
  \draw (62.498bp,120.0bp) node {$t$};
\end{scope}
\begin{scope}
  \definecolor{strokecol}{rgb}{1.0,0.0,0.0};
  \pgfsetstrokecolor{strokecol}
  \draw (148.5bp,143.0bp) -- (100.31bp,108.5bp) -- (196.69bp,108.5bp) -- cycle;
  \definecolor{strokecol}{rgb}{0.0,0.0,0.0};
  \pgfsetstrokecolor{strokecol}
  \draw (148.5bp,120.0bp) node {$T_4$};
\end{scope}
\begin{scope}
  \definecolor{strokecol}{rgb}{1.0,0.0,0.0};
  \pgfsetstrokecolor{strokecol}
  \draw (262.5bp,143.0bp) -- (214.31bp,108.5bp) -- (310.69bp,108.5bp) -- cycle;
  \definecolor{strokecol}{rgb}{0.0,0.0,0.0};
  \pgfsetstrokecolor{strokecol}
  \draw (262.5bp,120.0bp) node {$T_5$};
\end{scope}
\begin{scope}
  \definecolor{strokecol}{rgb}{0.0,0.0,0.0};
  \pgfsetstrokecolor{strokecol}
  \draw (19.5bp,23.0bp) ellipse (19.5bp and 19.5bp);
  \draw (19.498bp,23.0bp) node {$a_{i+1}$};
\end{scope}
\begin{scope}
  \definecolor{strokecol}{rgb}{1.0,0.0,0.0};
  \pgfsetstrokecolor{strokecol}
  \draw (105.5bp,46.0bp) -- (57.31bp,11.5bp) -- (153.69bp,11.5bp) -- cycle;
  \definecolor{strokecol}{rgb}{0.0,0.0,0.0};
  \pgfsetstrokecolor{strokecol}
  \draw (105.5bp,23.0bp) node {$T_6$};
\end{scope}
\end{tikzpicture}
\begin{tikzpicture}[scale=0.35,>=latex',line join=bevel]
  \pgfsetlinewidth{1bp}
\pgfsetcolor{black}
  \draw [] (314.5bp,172.82bp) .. controls (314.5bp,161.76bp) and (314.5bp,147.4bp)  .. (314.5bp,136.08bp);
  \draw [] (294.9bp,113.73bp) .. controls (249.12bp,108.9bp) and (132.63bp,92.596bp)  .. (48.498bp,46.0bp) .. controls (43.523bp,43.244bp) and (38.461bp,39.707bp)  .. (33.989bp,36.279bp);
  \definecolor{strokecol}{rgb}{0.0,0.0,0.0};
  \pgfsetstrokecolor{strokecol}
  \draw (126.5bp,71.5bp) node {$0$};
  \draw [] (295.36bp,111.65bp) .. controls (257.37bp,103.49bp) and (170.49bp,82.158bp)  .. (105.5bp,46.0bp) .. controls (100.53bp,43.235bp) and (95.468bp,39.695bp)  .. (90.996bp,36.267bp);
  \draw (179.5bp,71.5bp) node {$0$};
  \draw [] (296.54bp,108.78bp) .. controls (267.43bp,97.575bp) and (208.57bp,73.549bp)  .. (162.5bp,46.0bp) .. controls (157.62bp,43.081bp) and (152.58bp,39.495bp)  .. (148.11bp,36.073bp);
  \draw (225.5bp,71.5bp) node {$0$};
  \pgfsetcolor{blue}
  \draw [] (298.73bp,104.37bp) .. controls (279.59bp,90.8bp) and (246.42bp,67.065bp)  .. (218.5bp,46.0bp) .. controls (214.25bp,42.798bp) and (209.7bp,39.265bp)  .. (205.53bp,35.984bp);
  \definecolor{strokecol}{rgb}{0.0,0.0,0.0};
  \pgfsetstrokecolor{strokecol}
  \draw (265.5bp,71.5bp) node {$w$};
  \pgfsetcolor{red}
  \draw [] (305.69bp,98.716bp) .. controls (296.76bp,81.691bp) and (283.06bp,55.574bp)  .. (274.45bp,39.157bp);
  \definecolor{strokecol}{rgb}{0.0,0.0,0.0};
  \pgfsetstrokecolor{strokecol}
  \draw (297.5bp,71.5bp) node {$\ w_1$};
  \pgfsetcolor{red}
  \draw [] (323.3bp,98.716bp) .. controls (332.23bp,81.691bp) and (345.93bp,55.574bp)  .. (354.55bp,39.157bp);
  \definecolor{strokecol}{rgb}{0.0,0.0,0.0};
  \pgfsetstrokecolor{strokecol}
  \draw (343.5bp,71.5bp) node {$\ \ \ w_2$};
  \pgfsetcolor{red}
  \draw [] (330.74bp,105.18bp) .. controls (358.98bp,87.233bp) and (416.73bp,50.543bp)  .. (443.83bp,33.319bp);
  \definecolor{strokecol}{rgb}{0.0,0.0,0.0};
  \pgfsetstrokecolor{strokecol}
  \draw (397.5bp,71.5bp) node {$w_3$};
  \pgfsetcolor{red}
  \draw [] (332.54bp,109.11bp) .. controls (366.62bp,97.012bp) and (442.75bp,69.838bp)  .. (506.5bp,46.0bp) .. controls (518.82bp,41.393bp) and (532.58bp,36.088bp)  .. (543.49bp,31.841bp);
  \definecolor{strokecol}{rgb}{0.0,0.0,0.0};
  \pgfsetstrokecolor{strokecol}
  \draw (456.5bp,71.5bp) node {$\ \ \ w_4$};
  \pgfsetcolor{red}
  \draw [] (333.73bp,112.14bp) .. controls (382.35bp,103.42bp) and (513.93bp,78.48bp)  .. (620.5bp,46.0bp) .. controls (633.28bp,42.103bp) and (647.29bp,36.713bp)  .. (658.22bp,32.243bp);
  \definecolor{strokecol}{rgb}{0.0,0.0,0.0};
  \pgfsetstrokecolor{strokecol}
  \draw (556.5bp,71.5bp) node {$\ w_5$};
  \pgfsetcolor{red}
  \draw [] (334.11bp,113.87bp) .. controls (394.69bp,108.49bp) and (583.87bp,89.143bp)  .. (734.5bp,46.0bp) .. controls (747.35bp,42.32bp) and (761.36bp,36.933bp)  .. (772.27bp,32.411bp);
  \definecolor{strokecol}{rgb}{0.0,0.0,0.0};
  \pgfsetstrokecolor{strokecol}
  \draw (658.5bp,71.5bp) node {$\ w_6$};
\begin{scope}
  \definecolor{strokecol}{rgb}{0.0,0.0,0.0};
  \pgfsetstrokecolor{strokecol}
  \draw (314.5bp,191.0bp) ellipse (18.0bp and 18.0bp);
\end{scope}
\begin{scope}
  \definecolor{strokecol}{rgb}{0.0,0.0,0.0};
  \pgfsetstrokecolor{strokecol}
  \draw (314.5bp,116.5bp) ellipse (19.5bp and 19.5bp);
  \draw (314.5bp,116.5bp) node {$a_i$};
\end{scope}
\begin{scope}
  \definecolor{strokecol}{rgb}{0.0,0.0,0.0};
  \pgfsetstrokecolor{strokecol}
  \draw (19.5bp,23.0bp) ellipse (19.5bp and 19.5bp);
  \draw (19.498bp,23.0bp) node {$p$};
\end{scope}
\begin{scope}
  \definecolor{strokecol}{rgb}{0.0,0.0,0.0};
  \pgfsetstrokecolor{strokecol}
  \draw (76.5bp,23.0bp) ellipse (19.5bp and 19.5bp);
  \draw (76.498bp,23.0bp) node {$s$};
\end{scope}
\begin{scope}
  \definecolor{strokecol}{rgb}{0.0,0.0,0.0};
  \pgfsetstrokecolor{strokecol}
  \draw (133.5bp,23.0bp) ellipse (19.5bp and 19.5bp);
  \draw (133.5bp,23.0bp) node {$t$};
\end{scope}
\begin{scope}
  \definecolor{strokecol}{rgb}{0.0,0.0,0.0};
  \pgfsetstrokecolor{strokecol}
  \draw (190.5bp,23.0bp) ellipse (19.5bp and 19.5bp);
  \draw (190.5bp,23.0bp) node {$a_{i+1}$};
\end{scope}
\begin{scope}
  \definecolor{strokecol}{rgb}{1.0,0.0,0.0};
  \pgfsetstrokecolor{strokecol}
  \draw (266.5bp,46.0bp) -- (227.51bp,11.5bp) -- (305.49bp,11.5bp) -- cycle;
  \definecolor{strokecol}{rgb}{0.0,0.0,0.0};
  \pgfsetstrokecolor{strokecol}
  \draw (266.5bp,23.0bp) node {$T_1$};
\end{scope}
\begin{scope}
  \definecolor{strokecol}{rgb}{1.0,0.0,0.0};
  \pgfsetstrokecolor{strokecol}
  \draw (362.5bp,46.0bp) -- (323.51bp,11.5bp) -- (401.49bp,11.5bp) -- cycle;
  \definecolor{strokecol}{rgb}{0.0,0.0,0.0};
  \pgfsetstrokecolor{strokecol}
  \draw (362.5bp,23.0bp) node {$T_2$};
\end{scope}
\begin{scope}
  \definecolor{strokecol}{rgb}{1.0,0.0,0.0};
  \pgfsetstrokecolor{strokecol}
  \draw (458.5bp,46.0bp) -- (419.51bp,11.5bp) -- (497.49bp,11.5bp) -- cycle;
  \definecolor{strokecol}{rgb}{0.0,0.0,0.0};
  \pgfsetstrokecolor{strokecol}
  \draw (458.5bp,23.0bp) node {$T_3$};
\end{scope}
\begin{scope}
  \definecolor{strokecol}{rgb}{1.0,0.0,0.0};
  \pgfsetstrokecolor{strokecol}
  \draw (563.5bp,46.0bp) -- (515.31bp,11.5bp) -- (611.69bp,11.5bp) -- cycle;
  \definecolor{strokecol}{rgb}{0.0,0.0,0.0};
  \pgfsetstrokecolor{strokecol}
  \draw (563.5bp,23.0bp) node {$T_4$};
\end{scope}
\begin{scope}
  \definecolor{strokecol}{rgb}{1.0,0.0,0.0};
  \pgfsetstrokecolor{strokecol}
  \draw (677.5bp,46.0bp) -- (629.31bp,11.5bp) -- (725.69bp,11.5bp) -- cycle;
  \definecolor{strokecol}{rgb}{0.0,0.0,0.0};
  \pgfsetstrokecolor{strokecol}
  \draw (677.5bp,23.0bp) node {$T_5$};
\end{scope}
\begin{scope}
  \definecolor{strokecol}{rgb}{1.0,0.0,0.0};
  \pgfsetstrokecolor{strokecol}
  \draw (791.5bp,46.0bp) -- (743.31bp,11.5bp) -- (839.69bp,11.5bp) -- cycle;
  \definecolor{strokecol}{rgb}{0.0,0.0,0.0};
  \pgfsetstrokecolor{strokecol}
  \draw (791.5bp,23.0bp) node {$T_6$};
\end{scope}
\end{tikzpicture}
\centering
\end{figure}
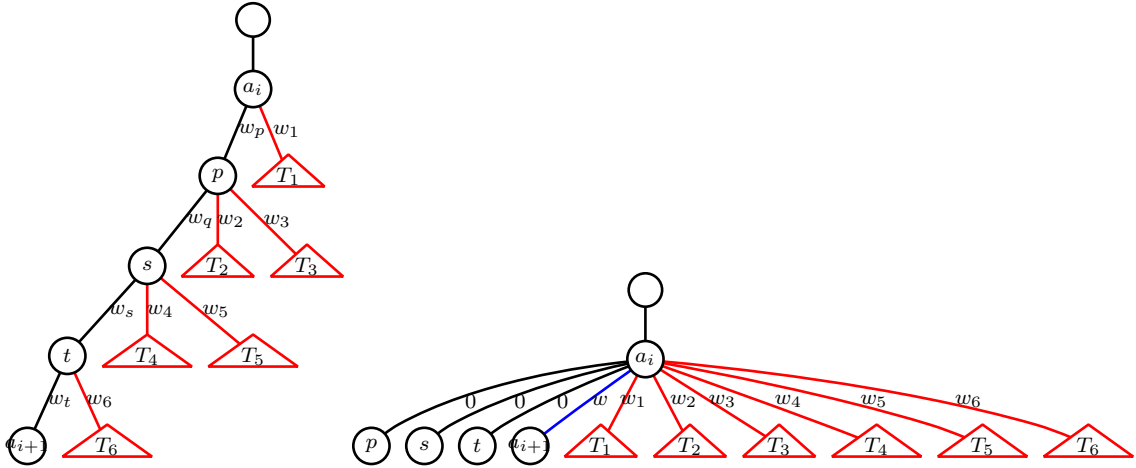

Each up-push acts on nodes in $P$ between two consecutive anchor points $a_i, a_{i+1}$ of the path $P$. During an up-push, we take all nodes in $P$ that lie between $a_i$ and $a_{i+1}$, which we will call $P'$, and make each node in $P'$ a child of $a_i$ with the edge connecting them to $a_i$ having weight 0. Suppose there is a child subtree $T_j$, which is a child of a node in $P'$ with edge connection cost $w_j$, the subtree $T_j$ will become a child of $a_i$ with the edge connecting them having cost $w_j$ (see Figure \ref{fig:upush}). 
Once we have completed up-pushes for all paths in $L_1$, we will find anchor points and perform up-pushes for each path in $L_2$. We will repeat this for paths in $L_i$ after our algorithm has finished up-pushes for paths in $L_{i-1}$.  

We will now describe how we can find the anchor points. We will first describe what we would like to achieve from anchor points. We want the cost associated with a path in $L_i$ for some tour to differ by at most $O(\eps)$ in our new tree compared to the the original tree. Suppose $P$ is a path
in $L_i$ and a tour $t$ is travelling $P$ down to node $u$ which is between
$a_i$ and $a_{i+1}$. Then the cost of the portion of the tour from root 
of $P$ to $a_i$ is the same in the original tree and the new tree; however
the cost to travel from $a_i$ to $u$ is zero. We would like this cost in the original tree to be a small factor of the cost from the root of $P$ to $a_i$.

Our algorithm works as follows from top to bottom. For any path $P$ in $L_i$, we will set the top node of the path to be $a_1$ and its child in $P$ to be $a_2$. Our goal is to pick $a_i$ and $a_{i+1}$ for $i > 2$ such that $w(a_i, a_{i+1}) > \eps \cdot w(a_1, a_i)$ and $w(a_i, v) \le \eps \cdot w(a_1, a_i)$ where $v$ is the last vertex on $a_i, a_{i+1}$ path before $a_{i+1}$. If there is no $a_{i+1}$ such that $w(a_i,a_{i+1}) > \eps \cdot w(a_1, a_i)$, then we set the last node of $P$ to be $a_{i+1}$. So, we pick $a_{i+1}$ to be the farthest vertex from $a_i$ in $P$ such that $w(a_i,v) \le \eps \cdot w(a_1, a_i)$ where $v$ is the last node before $a_{i+1}$. This in turn would imply that $w(a_1, a_{i+1}) > (1 + \eps) w(a_1, a_i)$, except if $a_{i+1}$ is the last node of the path. Hence,  $w(a_1, a_i) > (1 + \eps)^{i-2} w(a_1,a_2) > (1 + \eps)^{i-2}$. Since edge weights are at most $2n^3/\epsilon^2$, the number of anchor points are 
at most $\frac{\delta\log n}{\eps}$ for some constant $\delta>0$

\subsubsection{Analysis}
In the last section, we showed that every path in some level $L_i$ can be made to have at most $O\br{\frac{\log n}{\eps}}$ nodes.
\begin{lemma}
The height of the new tree is $O\br{\frac{\log^2 n}{\eps}}$. 
\end{lemma}
\begin{proof}
In our algorithm, we first decomposed the tree into a set of edge-disjoint paths. The decomposition guarantees that one would first visit a lower level node in any root-toeaf path before visiting one with a higher level. Since there are at most $O(\log n)$ different levels, any root-to-leaf path will be a disjoint union of paths from levels $L_1, \ldots, L_s$ and there can be at most one path from each level. Since the height of a path in any level, $L_i$ is at most $O\br{\frac{\log n}{\eps}}$, and there are at most $O(\log n)$ different levels, the maximum height in our new tree is at most $O\br{\frac{\log^2 n}{\eps}}$
\end{proof} 

Suppose we take a path $P$ at some level $L_c$. Let us fix a tour in an optimal solution and let the farthest point in $P$ the tour travels to be between anchor points [$a_i$,$a_{i+1}$). We use [$a_i$,$a_{i+1}$) denote that the tour crosses $a_i$ but will not cross $a_{i+1}$. Let $T$ be the original tree and let $T'$ be the new tree with reduced height. A tour in the optimal solution for $T'$ can visit nodes lying between $a_i$ and $a_{i+i}$ at no additional cost after visiting $a_i$. Suppose the cost of traversing the edges of $P$ in $T'$ is denoted by $d$, then the cost of traversing the edges of $P$ in $T$ is going to be at most  $(1 + O(\eps))d$. This is because the cost of the edges between $a_i$ and the vertex before $a_{i+1}$ sum to at most $O(\eps)w(r,a_i)$. Hence, the additional cost to cover them in $T$ is only going to be at most an $\eps$ fraction more. 
\begin{lemma}
Let $T$ be the original tree, $T'$ be the new tree, $\opt$ be the cost of the optimal set of tours covering $T$ and $\opt'$ be the cost of the optimal set of tours covering $T'$. Then, 
$$\opt' \le \opt \le (1 + \eps)\opt'.$$
\end{lemma}
\begin{proof}
Let us fix an optimal set of tours covering tree $T$ with cost $\opt$. Suppose we pick a tour $t$ and decompose this tour into paths each of which is entirely
within one level $L_i$. Suppose $P$ is a path of $t$ in some level $L_c$. Let the farthest point in $P$ the tour travels to be between anchor points [$a_i$,$a_{i+1}$). 
In our construction, the cost to visit any point lying between the root of $P$ and $a_i$ is the same in both $T$ and $T'$. However, in $T'$, the tour can visit any node lying between $a_i$ and $a_{i+1}$ for free, but the tour would have an additional cost to traverse these edges in tree $T$. Hence, for any path such $P$, the cost of a tour $t$ to traverse edges in $P$ is less in $T'$ compared to $T$. 
Since any tour costs no more in instance $T'$, we have $\opt' \le \opt$.

Conversely, the extra cost of covering points lying between $a_i$ and $a_{i+1}$ in $T$ is at most $O(\eps)$ times the cost of path $P$ (based on the property of anchor points). So the cost of using a path like $P$ is at most an $\eps$ factor more in $T$ compared to $T'$. Thus, the cost of any tour $t$ in $T$ is at most $1+\eps$ times the cost of the same tour in $T'$ and hence $\opt\leq (1+\eps)\opt'$
\end{proof} 

Instead of $T$, we can solve the instance on $T'$ with height $O(\log^2 n/\eps)$ and lift the solution for $T'$ back to a solution for $T$.  We obtain a solution for $T$ with cost at most $(1 + \eps)\opt$.

\section{QPTAS for Bounded Treewidth Graphs}
Given a graph $G = (V,E)$ with treewidth $k$, we will assume we are given a tree decomposition $T = (V',E')$. We will refer to $G$ as the graph and $T$ as the tree. We will refer to vertices in $V$ by \textbf{nodes} and vertices in $V'$ by \textbf{bags}. We will refer to edges in $E$ by \textbf{edges} and edges in $E'$ by \textbf{superedges}.
\begin{definition}
A \textbf{tree decomposition} of a graph $G$ is a pair $(T, \{B_t\}_{t \in V(T)})$, where $T$ is a tree whose every node $t \in V'$ is assigned a vertex subset $B_t \subseteq V(G)$, called a bag, such that the following three conditions hold: 
\begin{enumerate}
    \item $\cup_{t \in V(T)} B_t = V(G)$. In other words, every vertex of $G$ is in at least one bag. 
    \item For every $uv \in E(G)$, there exists a node $t$ of $T$ such that bag $B_t$ contains both $u$ and $v$. 
    \item For every $u \in V(G)$, the set $T_u = \{t \in V(T) : u \in B_t \}$, i.e., the set of nodes whose corresponding bags contain $u$, induces a connected subtree of $T$. 
\end{enumerate}
\end{definition} 
For a bag $s$, let $C_s$ denote the union of nodes in bags below $s$ including $s$. Bag $s$ forms a boundary or border between nodes in $C_s$ and $V(G) \setminus C_s$. We will assume an arbitrary bag containing the depot to be root of the tree decomposition. Let $k$ be the treewidth of our graph $G$. We will assume that following properties hold for our tree decomposition $T$ of $G$ from the work of Boedlander and Hagerup \cite{BodlaenderH95}, 
\begin{itemize}
    \item $T$ is binary. 
    \item $T$ has depth $O(\log n)$. 
    \item The width of $T$ is at most $k'=3k + 2$. 
\end{itemize}
To simplify  notation, by replacing $k'$ with $k$
we will assume $T$ has height $\delta \log n$ for some fixed $\delta > 0$ and each bag has width $k$. From the third property of a tree decomposition, we know that for every $u \in V(G)$, the set $T_u = \{t \in V(T) : u \in X_t \}$ i.e., the set of nodes whose corresponding bags contain $u$, induces a connected subtree of $T$. Since the bags associated with a node $u \in V(G)$ correspond to a subtree in $T$, we will place the demand/tokens of $u$ at the root bag of the tree $T_u$ i.e. the bag containing $u$ closest to the root bag of $T$. Since $T_u$ is a tree, we are guaranteed a unique root bag of $T_u$ exists. We are doing this to ensure that the demand of a client is delivered exactly once. 

Similar to how we showed the existence of a near-optimum solution for trees, we will modify the optimum solution $\OPT$ in a bottom-up manner by modifying the tours covering the set of nodes below bag $s$, $C_s$. For each bag $s$, we change the structure of the partial tours going down $C_s$ (by adding a few extra tours from the depot) and also adding some extra tokens for nodes in bag $s$ so that the partial tours that visit $C_s$ all have a size from one of polyogarithmic many possible sizes (buckets) while increasing the number and the cost of the tours by a small factor. Note that although a node can be in different bags, its initial demand is in one bag and we might add extra tokens to copies of it in other bags. 

Similar to the case of tree, we assume the bags of the tree decomposition are partitioned into levels $V_1,\ldots,V_h$ where $V_1$ is the bag containing the depot and $h$ is the height of $T$.  For every tour $\calT$ and every level $\ell$, we can define the notion of top and bottom part similar to the case of trees. For every $C_s$, a tour $\calT$ enters $C_s$ through bag $s$ using a node $x$ and exists through node $z$ where both $x$ and $z$ have to be in $s$. Note that $x$ and $z$ could be equal if the tour enters and exists $s$ using the same node. For a bag $s$, let $n_{s}^{x,z}$ be the number of partial tours covering nodes in $C_s$ that enter through $x$ and exit through $z$ in $s$. For each bag and entry/exit pair, we will define the notion of a small/big bucket similar to the case of trees. For a big bucket, we will place the $n_{s}^{x,z}$ tours (ordered by increasing size) into groups $G_1^{x,z,s}, \ldots, G_g^{x,z,s}$ of equal sizes. Let $h_{i}^{s,x,z, \max} (h_{i}^{s,x,z, \min}) $ refer to the maximum (minimum) size of the tours in $G^{x,z,s}_i$.

Similar to the case of trees, let $f$ be a mapping from a tour in $G_i^{x,z,s}$ to one in $G_{i-1}^{x,z,s}$. Now suppose we modify $\OPT$ to $\OPT'$ in the following way: for each tour $\calT$ that has a partial tour in $t \in G_i^{x,z,s}$, replace the bottom part of $\calT$ entering through $x$ and exiting through $z$ in $s$ from $t$ to $f(t)$ (which is in $G_{i-1}^{x,z,s}$). The only problem is that those tokens in $C_s$ that were picked up by the partial tours in $G_g^{x,z,s}$ are not covered by any tours and like the case of trees, these are \emph{orphant} tokens. For each tour $\calT$ and its (new) partial tour $t \in G_i^{x,z,s}$, if we add $h_i^{{x,z,s}, \max} - |t|$ extra tokens at $s$ to be picked up by $t$, then each partial tour has size exactly same as the maximum size of its group without violating the capacities. 
Similar to the case of trees, we will show that if $n_{s}^{x,z}$ is sufficiently large (at least polylogarithmic), then if we sample a small fraction of the tours of the optimum at random and add two copies of them (as extra tours), they can be used to cover the orphant tokens. 

\subsection{Changing $\OPT$ to a near-optimum structured solution}
Similar to the structure theorem for trees, we will modify the optimal solution $\OPT$ to a near-optimum solution $\OPT'$ having certain properties. We will start at the last level, and modify partial tours from $\OPT$ at level $\ell$ to obtain $\OPT_{\ell}$. We will then iteratively obtain $\OPT_{\ell -1}$ by modifying partial tours from $\OPT_\ell$ at level $\ell - 1$, and iteratively do this for each level until we obtain $\OPT_1 = \OPT'$. 

\begin{definition}
For a bag $s$, the $i$-th \text{bucket}, $b_i$, entering at $x$ and exiting at $z$ contains the number of tours of $\OPT_\ell$ having coverage between $[\sigma_i, \sigma_{i+1})$ tokens in $C_s$ where $\sigma_i$ is the $i$-th threshold value. We will denote this by a entry/exit-bag-bucket configuration $(s,b_i,x,z)$. Let $n_{s,i}^{x,z}$ be the number of tours in bucket $b_i$ entering through $x$ and exiting through $z$ in bag $s$. 
\end{definition}

\begin{definition}
An entry/exit-bag-bucket configuration $(s,b_i,x,z)$ is \textbf{small} if $n_{s,i}^{x,z}$ is at most $\alpha \log^2 n / \eps$ and is \textbf{big} otherwise,
for a constant $\alpha \ge \max\{1, 20\delta\}$.
\end{definition}

Note that for any bag $s$ and entry/exit-bag-bucket configuration $(s,b_i,x,z)$, if $(s,b_i,x,z)$ is small, we do not modify the partial tours in it. 
However, if $(s,b_i,x,z)$ is a big bucket, we create groups $G_{i,1}^{s,x,z}, \ldots, G_{i,g}^{s,x,z}$ of equal sizes, for $g = (2\delta \log n)/\eps$; so 
$|G_{i,j}^{s,x,z}|=\lceil n_{s,i}^{x,z}/g\rceil$. We also consider a mapping $f$ (as before) which maps (in the same order) the 
tours $t\in G^{s,x,z}_{i,j}$ to the tours in $G^{s,x,z}_{i,j-1}$ for all $1<j\leq g$. 
Consider set ${\bf T}_\ell$ of all the tours $\calT$ in $\OPT_\ell$ that visit a bag in one of the lower levels $V_{\geq \ell}$. Consider an arbitrary such tour $\calT$ that has a partial tour $t$ in a big entry/exit-bag-bucket configuration $(s,b_i,x,z)$, 
suppose $t$ belongs to group $G_{i,j}^{s,x,z}$. We replace $t$ with $f(t)$ in $\calT$.

Now, add some extra tokens at $x$ to be picked up by $\calT$ so that
the size of the partial tour of $\calT$ at $C_s$ is exactly $h_{i,j-1}^{s,x,z, \max}$.
If we make this change for all tours $\calT\in {\bf T}_\ell$, each partial tour of them at level $\ell$ that was in a group
$j<g$ of a big entry/exit-bag-bucket configuration $(s,b_i,x,z)$ is replaced with a smaller partial tour from group $j-1$ of the same big entry/exit-bag-bucket configuration; after adding extra tokens to $x$ at bag $s$ (if needed), the size is the maximum size from group $j-1$.  The tokens that were picked by
partial tours from $G_{i,g}^{s,x,z}$ for a big entry/exit-bag-bucket configuration $(s,b_i,x,z)$ are now orphant.
We are going to (randomly) select a subset of tours of $\OPT$ as "extra tours" and add them to $\OPT'$ and modify them
such that they cover all the tokens that are now orphant (i.e. those that were covered by partial tours of $G_{i,g}^{s,x,z}$ for all big entry/exit-bag-bucket configuration $(s,b_i,x,z)$ at level $\ell$).
Suppose we select each tour $\calT$ of $\OPT$ with probability $\eps$. We make two copies of the \emph{extra tour} and we designate both extra copies to bags at one of the levels $V_\ell$ that it visits with equal probability.

\begin{lemma}\label{lem:cost-btree}
The expected cost of extra tours selected is $2\eps \cdot \opt$. 
\end{lemma}
\begin{proof}
Suppose $f^+(e)$ and $f^-(e)$ denote the number of tours traveling edge $e$ in each of the two directions.
So the contribution of edge $e$ to the optimal solution is $2\cdot w(e)\cdot (f^+(e) + f^-(e))$; $\opt = \sum_{e \in E} w(e) \cdot (f^+(e) + f^-(e))$. Let $m^{+}(e)$ ($m^{-}(e)$) denote the number of sampled tours from the tours contributing to $f^+(e)$ ($f^-(e)$). Since we used two extra copies for each sampled tour, the number of extra tours for an edge $e$ is $2(m^{+}(e) + m^{-}(e))$. Let $\calT_{e,1}, \ldots, \calT_{e,f^+(e) + f^-(e)}$ be the tours using $e$ in either directions. Like in the case of trees, it is possible for a tour to use edge $e$ in both directions.  Let $Y_{e,i}$ be a random variable which is 1 if tour $\calT_{e,i}$ is sampled and $0$ otherwise. 
$$\Ex{Y_{e,i}} = \Prob{\calT_{e,i} \text{ is sampled}} =  \eps.$$
Let $m^{+}(e) + m^{-}(e) = Y_e = \sum_{i = 1}^{f^+(e) + f^-(e)} Y_{e,i}$. By linearity of expectations, we have 
$$\Ex{m^{+}(e) + m^{-}(e)} = \Ex{Y_e} = \sum_{i = 1}^{f^+(e) + f^-(e)} \Ex{Y_{e,i}} = \sum_{i =1}^{f^+(e) + f^-(e)} \eps = \eps \cdot (f^+(e) + f^-(e)).$$
Summing up the extra cost over all edges, the expected cost of the extra tours is 
$$2 \sum_{e \in E} \Ex{m^{in}(e) + m^{out}(e)} = 2\eps \cdot \sum_{e \in E}(f^+(e) + f^-(e)) = 2\eps \cdot \opt.$$
\end{proof}

Therefore, we can assume that the expected cost of all extra tours added is at most $2 \eps \cdot \opt$. Let $X_\ell$ be the set of extra tours designated to bags in level $\ell$. We assume we add $X_\ell$ when we are building $\OPT_\ell$ (it is only for the sake of analysis). For each bag $s \in V_\ell$ and entry/exit-bag-bucket configuration $(s,b_i,x,z)$, let $X^{s,x,z}_{i}$ be those in $X_\ell$ whose partial tour in $C_s$ has a size in bucket $b_i$. Each extra tour in $X_\ell$ will not be picking any of the tokens in levels $V_{< \ell}$ (as they will be covered by the tours already in $\OPT_\ell)$; they are used to cover the orphant tokens created by partial tours of $G_{i,g}^{s,x,z}$ for each big entry/exit-bag-bucket configuration $(s,b_i,x,z)$ with $s \in V_\ell$; as described below. 

\begin{lemma}\label{lem:extra-btree}
For each level $V_\ell$, each bag $s\in V_\ell$ and big entry/exit-bag-bucket configuration $(s,b_i,x,z)$,
w.h.p. $|X^{s,x,z}_{i}|\geq \frac{\eps^2}{\delta\log n}\cdot n_{s,i}^{x,z}$.
\end{lemma}
\begin{proof}
Suppose $(s,b_i,x,z)$ is a big entry/exit-bag-bucket configuration at some level $V_\ell$. Let $p_1, \ldots, p_{n_{s,i}^{x,z}}$ be the partial tours in the entry/exit-bag-bucket configuration $(s,b_i,x,z)$. Let the tour in $\OPT$ corresponding to $p_i$ be $\calT$. Two copies of tour $p_i$ are assigned to $b_i$ if both of the following events are true: 
\begin{itemize}
    \item Let $A_i$ be the event where tour $\calT$ is sampled as an extra tour. Since each tour is sampled with probability $\eps$, we have $\Prob{A_i} = \eps$. 
    \item Let $B_i$ be the event where tour $\calT$ is assigned to level $\ell$. There are $h = \delta \log n $ many levels and since $\calT$ (if sampled) is assigned to any one of its levels, $\Prob{B_i} \ge 1/h \ge 1/(\delta \log n)$. 
\end{itemize}
Let $Y_i$ be a random variable which is 1 if $p_i$ is an extra tour in $(v,b_i)$ and 0 otherwise. 
$$\Ex{Y_i} = \Prob{Y_i = 1} = \Prob{A_i \land B_i}  = \Prob{A_i} \cdot \Prob{B_i} \ge \eps/ (\delta \log n).$$
Let $Y^{s,x,z}_{i} = \sum_{i = 1}^{n_{s,i}^{x,z}} Y_i$ be the random variable keeping track of the number of sampled tours in $(s,b_i,x,z)$. The number of extra tours, $|X^{s,x,z}_{i}| = 2Y^{s,x,z}_{i}$ since we add two copies of a sampled tour to $X^{s,x,z}_{i} $. By linearity of expectation, we have 
$$\Ex{|X^{s,x,z}_{i}|} = 2\Ex{Y^{s,x,z}_{i}} = 2\sum_{i = 1}^{n_{s,i}^{x,z}} \Ex{Y_i} \ge \frac{2\eps}{\delta \log n} \cdot n_{s,i}^{x,z}. $$
We want to show that $|X^{s,x,z}_{i}| \ge \frac{\Ex{|X^{s,x,z}_{i} |}}{2} \ge \frac{\eps}{\delta \log n} \cdot n_{s,i}^{x,z} $ with high probability over all vertex-bucket pairs. 

Using Chernoff Bound with $\mu = \Ex{|X^{s,x,z}_{i}|} \ge \frac{2\eps^2}{\delta \log^2 n} \cdot n_{s,i}^{x,z} \ge 24 \log n$ since $n_{s,i}^{x,z} \ge \alpha \log^2n /\eps$ and $\alpha \ge 20\delta$. 
\begin{equation*}
    \begin{split}
        \Prob{|X^{s,x,z}_{i} | < \frac{\Ex{|X^{s,x,z}_{i} |}}{2}} & \le e^{-(5\log n)} = \frac{1}{n^5}
    \end{split}
\end{equation*}
Note that the above equation only shows the concentration bound for a single entry/exit-bag-bucket configuration. For a bag, there are $O(k^2)$ many entry/exit pairs. There are $O(kn)$ bags and $\tau = O(\log n /\eps)$ buckets, so the total number of entry/exit-bag-bucket configuration is at most $O(k^2n \log n/\eps)$. Suppose we do a union bound over all buckets, we get
$$\sum_{\text{all }  (s,b_i,x,z) \text{ configurations}} \Prob{|X^{s,x,z}_{i}| < \frac{\Ex{|X^{s,x,z}_{i}|}}{2}} \le \frac{1}{n}. $$
We showed that for every entry/exit-bag-bucket configuration $(s,b_i,x,z)$, $|X^{s,x,z}_{i}|\geq \frac{\eps}{\delta\log n}n_{s,i}^{x,z}$ holds with high probability. 
\end{proof}

\begin{lemma}\label{lem:extra2-bree}
Consider all bags $s\in V_\ell$, big entry/exit-bag-bucket configuration $(s,b_i,x,z)$ and the partial tours in $G_{i,g}^{s,x,z}$.
We can modify the tours in $X^{s,x,z}_{i}$ (without increasing the cost) and adding some extra tokens at nodes in $s$ (if needed) so that:
\begin{enumerate}
\item The tokens picked up by partial tours in $G_{i,g}^{s,x,z}$ are covered by some tour in $X^{s,x,z}_{i}$, and
\item The new partial tours that pick up the orphant tokens in $G_{i,g}^{s,x,z}$ 
have size exactly $h_{i,g}^{s,x,z, \max}$ and all tours still have size at most $Q$.
\item For each (new) partial tour of $X^{s,x,z}_i$ and every level $\ell'>\ell$, the size of partial tours of $X^{s,x,z}_{i}$ at a
bag $s'$ at level $\ell'$ is also one of $O((\log Q \log^2 n)/\eps^2)$ many possible sizes.
\end{enumerate}
\end{lemma}
\begin{proof}
Our proof is going to be very similar to Lemma \ref{lem:extra2} for the case of trees. 
Our goal is to use the extra tours in $X^{s,x,z}_{i}$ to cover tokens picked up by partial tours of $G_{i,g}^{s,x,z}$ and we want each extra tour in $X^{s,x,z}_{i}$ to cover exactly $h_{i,g}^{s,x,z, \max}$ tokens. The tours in the last group, $G_{i,g}^{s,x,z}$, cover $\sum_{t \in G_{i,g}^{s,x,z}} |t|$ many tokens. We will add $\sum_{t \in G_{i,g}^{s,x,z}} (h_{i,g}^{s,x,z, \max} - |t|)$ extra tokens in node $x$ at bag $s$ for each entry/exit-bag-bucket configuration $(s,b_i,x,z)$ so that there are $h_{i,g}^{s,x,z, \max}$ tokens corresponding to each partial tour in $G_{i,g}^{s,x,z}$.  From now on, we will assume each partial tour in a last group $G_{i,g}^{s,x,z}$ covers $h_{i,g}^{s,x,z, \max}$ tokens.

Using Lemma \ref{lem:extra-btree}, we know with high probability that $|X^{s,x,z}_{i}|/|G_{i,g}^{s,x,z}| \ge 2$ since $|X^{s,x,z}_{i}| \ge \frac{\eps }{\delta\log n} \cdot n_{s,i}^{x,z} =2 |G_{i,g}^{s,x,z}|$.  Let $Y^{s,x,z}_{i} $ denote the number of tours in entry/exit-bag-bucket configuration $(s,b_i,x,z)$ that were sampled, so $|X^{s,x,z}_{i}| = 2|Y^{s,x,z}_{i} |$ and $|Y^{s,x,z}_{i} |\ge |G_{i,g}^{s,x,z}|$ with high probability. We will start by creating a one-to-one mapping $s : G_{i,g}^{s,x,z} \rightarrow Y^{s,x,z}_{i} $ which maps each tour in $G_{i,g}^{s,x,z}$ to a sampled tour in $Y^{s,x,z}_{i} $. We know such a one-to-one mapping exists since $|Y^{s,x,z}_{i} |\ge |G_{i,g}^{s,x,z}|$. 

Let $\calT$ be a sampled tour in $Y^{s,x,z}_{i} $ with two extra copies of it, $\calT_1$ and $\calT_2$ in $X^{s,x,z}_{i}$. Let the partial tours of $\calT$ at the bottom part in $V_\ell$ be $p_1, \ldots, p_m$. We know $|\calT| \ge \sum_{i = 1}^m |p_i|$. Like the case for trees, $s$ maps at most one tour in $G_{i,g}^{s,x,z}$ to each $p_j$. If a tour from $G_{i,g}^{s,x,z}$ maps to $p_j$, we will assume the load assigned to $p_j$ would be $r_j = h_{i,g}^{s,x,z, \max}$ and $p_j$ has load 0 if no tour is assigned to it.  

Suppose we think of $r_1, \ldots, r_m$ as items and $\calT_1$ and $\calT_2$ as bins of size $Q$. We might not be able to fit all items $r_1, \ldots, r_m$ into a bin of size $Q$ because $\sum_{i = 1}^m|r_i| \le (1 + \eps)\sum_{i = 1}^m |p_i| \le (1 + \eps)|\calT| \le (1 + \eps)Q$. Similar to the case of trees, we can show that we can assign $r_1, \ldots, r_j$ (for the maximum $j$) to $\calT_1$ such that $\sum_{i = 1}^j |r_i| \le Q$ and the rest, $r_{j+1}, \ldots, r_m$ to $\calT_2$ such that both $\calT_1$ and $\calT_2$ cover at most $Q$ tokens and all items $r_1, \ldots, r_m$ are covered by either $\calT_1$ or $\calT_2$. Hence, we have shown that the extra partial tours pick up exactly $h_{i,g}^{s,x,z, \max}$ while picking up orphant tokens from $G_{i,g}^{s,x,z}$. 

Also, the size of the extra tours after this modification at each bag $s'$ at any level $\ell'>\ell$ is essentially the same as what each of $r_i$'s were at those levels and since we go bottom to top in the tree, each of those partial tours $r_i$ have a size that either
belongs to a small bucket (and hence has one of $\alpha \log^2 n/\eps$ many sizes) or a big entry/exit-bag bucket (and hence has one of $O((\log Q\log n)/\eps^2)$ many sizes). 
Therefore, the size of partial tours of $X^{s,x,z}_i$ at any bag $s'$ at level $\ell'>\ell$ is one of $O((\log Q\log^2 n)/\eps^2)$ many sizes.
\end{proof} 

Therefore, using Lemma \ref{lem:extra2-bree}, all the tokens of $C_s$ remain covered by partial tours;
those partial tours in $G_{i,j}^{s,x,z}$ (for $1\leq j<g$) are tied to the top parts of the tours from group $G_{i,j+1}^{s,x,z}$ and the
partial tours of $G_{i,g}^{s,x,z}$ will be tied to extra tours designated to level $\ell$. We also add extra tokens at nodes in $s$ to be
picked up by the partial tours of $C_s$ so that each partial tour has a size exactly equal to the maximum size of a group.
All in all, the extra cost paid to build $\OPT_\ell$ (from $\OPT_{\ell+1}$) is for the extra tours designated to level $\ell$.

\begin{theorem}\label{lem:struct2} \textbf{(Structure Theorem)}
Let $\opt$ be the cost of the optimal solution to instance $\calI$. We can build an instance $\calI'$
such that each node has $\geq 1$ tokens and there exists a near-optimal solution $\OPT'$ for $\calI'$
having expected cost $(1 + 2\eps)\opt$ with the following property.
The partial tours going down $C_s$ for every bag $s$ in $\OPT'$ has one of $O((\log Q \log^2 n)/\eps^2)$ possible sizes.
More specifically, suppose $(s,b_i,x,z)$ is a entry/exit-bag-bucket configuration for $\OPT'$. Then either:
\begin{itemize}
\item $b_i$ is a small bucket and hence there are at most $\alpha \log^2 n/\eps$ many partial tours of $C_s$ whose size is in bucket $b_i$, or
\item $b_i$ is a big bucket; in this case there are $g = (2\delta \log n)/\eps $ many group sizes in $b_i$: 
$\sigma_i \leq h_{i,1}^{s,x,z,max}\leq \ldots\leq h_{i,g}^{s,x,z,max}<\sigma_{i+1}$ and every tour of bucket $i$ has one of these sizes. 
\end{itemize}
\end{theorem}

\begin{proof}
We will show how to modify $\OPT$ to a near-optimal solution $\OPT'$. We start from $\ell = h$ and let $ \OPT_\ell = \OPT$. For decreasing values of $\ell$ we show, for each $\ell$ how to modify $\OPT_{\ell+1}$ to obtain $\OPT_\ell$. We do this in the following manner: we do not modify partial tours in small entry/exit-bag-bucket configuration. However, for tours in big entry/exit-bag-bucket configuration $(s,b_i,x,z)$ in level $\ell-1$, we place them into $g$ groups $G_{i,1}^{s,x,z},\ldots,G_{i,g}^{s,x,z}$ of equal sizes by placing the $i$'th $n_{s,i}^{x,z}/g$ partial tours into $G_{i,j}^{s,x,z}$. We have a mapping $f$ from each partial tour in $G_{i,j-1}^{s,x,z}$ to one in $G_{i,j}^{s,x,z}$ for $j \in \{2, \ldots, g\}$. We modify $\OPT_\ell$ to $\OPT_{l+1}$ in the following way: for each tour $\calT$ that has a partial tour $t \in G_{i,j}^{s,x,z}$, replace the bottom part of $\calT$ at $s$ from $t$ to $f(t)$ (which is in $G_{i,j-1}^{s,x,z}$). For each tour $t \in G_{i,j-1}^{s,x,z}$, we will add $h_{i,j-1}^{s,x,z, \max} - |t|$ many extra tokens at $x$ in $s$. Note that by this change, the size of any tour such as $\calT$ can only decrease and we are not violating feasibility of the tour because $h_{i,j-1}^{s,x,z, \max} \le h_{i,j}^{s,x,z, \min}$. However, the tokens in $C_s$ picked up by the partial tours in $G_{i,g}^{s,x,z}$ are not covered by any tours. We can use Lemma \ref{lem:extra2-bree} to show how we can use extra tours to cover the partial tours in $G_{i,g}^{s,x,z}$ such that the new partial tours have size exactly $h_{i,g}^{s,x,z, \max}$.  

We will inductively repeat this for levels $\ell-2, \ell-3, \ldots, 1$ and obtain $\OPT_1 = \OPT'$. Note that by adding extra tokens $h_{i,j-1}^{s,x,z, \max} - |t|$ for a tour $t \in  G_{i,j-1}^{s,x,z}$, we are enforcing that the coverage of each tour is the maximum size of tours in its group. In a big bucket, there are $g = (2\delta \log n)/\eps$ many group sizes, so there are $O(\log n/\eps)$ possible sizes for tours in big entry/exit-bag-bucket configuration at a node. In a small entry/exit-bag-bucket configuration, there can be at most $\alpha \log^2 n/\eps$ many tours and since there are $\tau = O(\log Q /\eps)$ many buckets, there can be at most $O((\log Q \log^2 n)/\eps^2)$ many tour sizes covering $C_b$.  

Using Lemma \ref{lem:cost-btree}, we know the expected cost of the extra tours is at most $2\eps \cdot \opt$, so the expected cost of $\opt' 
\le (1 + 2\eps) \opt$.
\end{proof}

\subsection{Dynamic Program}\label{sec:tw-dp}
In this section we prove Theorem \ref{thm:treewidth} by presenting
a dynamic program that will compute a near optimum solution guaranteed 
by the structure theorem (Theorem \ref{lem:struct2}).
For a given bag $s$, we will estimate the number of tours entering and exiting $s$. Informally, we will have a vector $\nvec^{s,x,z} \in [n]^\tau$ where if $i < 1/\eps$, $\nvec_{i}^{s,x,z}$ keeps track of the exact number of tours covering $i$ tokens in $C_s$ by entering through $x$ and exiting though $z$ and if $i \ge 1/\eps$, $\nvec_{i}^{s,x,z}$ keeps track of the number of tours covering between $[\sigma_i, \sigma_{i+1})$ tokens. Let $a_s$ denote the total number of tokens to be picked up from nodes from bags below and including bag $s$. Since each bag $s$ has $k$ nodes, we use $\ovec_s \in [n]^{k}$ to denote the extra tokens to be picked up from nodes at bag $s$. If $v$ is a node in bag $s$, then $\ovec_{s,v}$ denotes the number of extra tokens to be picked up at $v$ in $s$. For a given entry/exit-bag-bucket configuration $(s,b_i,x,z)$, we will keep track of other pieces of information conditional on whether it is small or big. If entry/exit-bag-bucket configuration $(s,b_i,x,z)$ is small, we will store all tour sizes exactly. Since the number of tours in a small entry/exit-bag-bucket configuration is at most $\gamma = \alpha \log^2n /\eps$, we will use a vector $\tvec^{s,x,z,i} \in [n]^\gamma$ to represent the tours where  $\tvec^{s,x,z,i}_j$ represents the size of the $j$-th tour in the $i$-th bucket of tours covering $C_s$ entering through $x$ and exiting through $z$. 

If the entry/exit-bag-bucket configuration $(s,b_i,x,z)$ is big, there are $g = (2\delta \log n)/\eps$ many tour sizes corresponding to $n^{O(g)}$ possibilities. For each  entry/exit-bag-bucket configuration $(s,b_i,x,z)$, we need to keep track of the following information, 
\begin{itemize}
    \item $\hvec^{s,x,z,i} \in [n]^g$ is a vector where $\hvec^{s,x,z,i}_j = h_{i,j}^{s,x,z, \max}$, which is the size of the maximum tour which lies in group $G_{i,j}^{s,x,z}$ of bucket $i$ at bag $s$ entering through $x$ and exiting through $z$. 
    \item $\lvec^{s,x,z,i} \in [n]^g$ is a vector where $\lvec^{s,x,z,i}_j$ denotes the number of partial tours covering $h_{i,j}^{s,x,z, \max}$ tokens which lies in group $G_{i,j}^{s,x,z}$ of bucket $i$ at bag $s$ entering through $x$ and exiting through $z$. 
\end{itemize}
For a bag $s$ and entry/exit pairs, let $\pvec_{s,x,z}$ be a vector containing information about all tours entering and exiting $s$ through $x$ and $z$ across all buckets.
$$\pvec_{s,x,z} = [\nvec^{s,x,z}, (\tvec^{s,x,z,1},\hvec^{s,x,z,1},\lvec^{s,x,z,1}), (\tvec^{s,x,z,2},\hvec^{s,x,z,2},\lvec^{s,x,z,2}), \ldots, (\tvec^{s,x,z,\tau},\hvec^{s,x,z,\tau},\lvec^{s,x,z,\tau}) ]. $$
Similar to the case of trees, an entry/exit-bag-bucket configuration $(s,b_i,x,z)$ is either small or big and cannot be both, hence given $(\tvec^{s,x,z,i},\hvec^{s,x,z,i},\lvec^{s,x,z,i})$, it cannot be the case that $\tvec^{s,x,z,i} \neq \vec{0}, \hvec^{s,x,z,i} \neq \vec{0}$ and $\lvec^{s,x,z,i} \neq \vec{0}$. Since a bag $s$ contains $O(k)$ nodes, then we will let $\yvec_s$ denote a configuration of all partial tours covering tokens in $C_s$ which are entering and exiting $s$. Let $v_1, \ldots, v_d$ be the set of all nodes in $s$, then $\yvec_s$ contains information of tours entering and exiting $s$ through pairs of nodes in
$\{v_1, \ldots, v_d\}$. Note that a tour can enter and exit $s$ through the same node. 
$$\yvec_s = [a_s, \ovec_s, \pvec_{s,v_1,v_1}, \pvec_{s,v_1,v_2}, \ldots, \pvec_{s,v_{d}, v_{d-1}}, \pvec_{s, v_d, v_d}]. $$
The subproblem $\abold[s, \yvec_s]$ is supposed to be the minimum cost collection of partial tours covering $C_s$ having tour profiles corresponding to $\yvec_s$. Our dynamic program heavily relies on the properties of the near-optimal solution characterized by the structure theorem. We will compute $\abold[\cdot, \cdot]$ in a bottom-up manner, computing $\abold[s, \yvec_s]$ after we have computed entries for the children bags of $s$. 

The final answer is obtained by looking at various entries of the root bag of the tree decomposition, denoted by $r_s$. We will take the minimum cost entry amongst $\abold[r_s, \yvec_{r_s}]$ such that $\yvec_{r_s}$ is the configuration where all tours enter and exit $r_s$ only through the depot, $r$. We will compute our solution in a bottom-up manner. 

For any nodes $u,v$ in bag $s$, if there is no edge between $u$ and $v$, we can add an edge between them and the cost of the edge is the shortest path cost between $u$ and $v$ in $G$. Similarly, for two adjacent bags, $s$ and $s_1$, if $u \in s$ and $v \in s_1$ and if there is no edge between $u$ and $v$ in $G$, we will add an edge between them and the cost of the edge is the shortest path cost between $u$ and $v$ in $G$. If $u = v$, then the cost of the edge connecting them can be assumed to be zero. Let $\norm{\ovec_s} = \sum_{u \in s}\ovec_{s,u}$. 

For the base case, we consider leaf bags. A leaf bag $s$ could have $a_s \ge 1$ tokens where $a_s = \norm{\ovec_s}$. We will defer how we compute $\abold[s, \yvec_s]$ to the end of this section. Informally, we will set $\abold[s, \yvec_s]$ to be the minimum cost of the edges between nodes in bag $s$ used for the tours in $\yvec_s$ to pick up $\ovec_s$ tokens located at nodes in bag $s$. The total capacity of the tours in $\yvec_s$ should be exactly $a_s$ and a token at a node should be picked up by one of the tours in $\yvec_s$.  From our structure theorem, we know there exists a near optimum solution such that each partial tour has one of $O(\log Q \log^2 n / \eps^2)$ tour sizes and for each small bucket, there are at most $\alpha \log^2 n/ \eps$ partial tours in it. For every big bucket, there are $g = (2 \delta \log n)/\eps$ many group sizes and every tour of bucket $i$ has one of those sizes. We are computing all possible $\abold[s, \yvec_s]$ entries and from our structure theorem, we know one of them has near-optimum expected cost, so by enumerating all possibilities, our dynamic program finds a near-optimums solution for the leaf bag, proving the base case.

Recall that the tree $T$ is binary. Suppose bag $s$ has two children in $T$, $s_1$ and $s_2$. To compute cell $\abold[s, \yvec_s]$, we will use the entries of its children, $\abold[s_1, \yvec']$ and $\abold[s_2, \yvec'']$. Suppose $C_{s_i}$ has $a_{s_i}$ tokens, then $a_s = \norm{\ovec_s} + a_{s_1} + a_{s_2}$. $\hbold[\ovec_s, \yvec_s, \yvec', \yvec'']$ checks whether the tour profiles $\yvec_s, \yvec'$ and $\yvec''$ are consistent meaning that all tokens picked up by tours in $\yvec'$ and $\yvec''$ along with tokens in $s$, $\ovec_s$ are picked up by tours in $\yvec_s$. We will also define $\ibold[\cdot, \cdot, \cdot,\cdot]$ where $\ibold[\ovec_s, \yvec_s, \yvec', \yvec'']$ denotes the cost of using the edges in bag $s$, edges connecting nodes in $s$ and $s_1$, and edges connecting nodes in $s$ and $s_2$. We can think of $\ibold$ as the cost of using edges to patch up partial tours covering $C_{s_1}$ and partial tours covering $C_{s_2}$ to create tours covering $C_s$. We will explain in the next section how $\hbold$ and $\ibold$ are computed.  Recall $\ovec_s$ is part of $\yvec_s$. Suppose we have already computed the entries $\abold[s_1, \cdot]$ and $\abold[s_2, \cdot]$, we will compute $\abold[s, \cdot]$ in the following way:  
$$ \abold[s, \yvec_s] = \mi{\yvec', \yvec'' :\hbold[\ovec_s, \yvec_s, \yvec', \yvec''] = \text{True}}\{ \abold[s_1, \yvec'] +\abold[s_2, \yvec''] + \ibold[\ovec_s, \yvec_s, \yvec', \yvec''] \}.$$
There are four possibilities for each partial tour $t$ at bag $s$ going down $C_s$ covering tokens for the subtree rooted at children bags, $s_1$ and $s_2$ while also picking up extra tokens from nodes in $s$: 
\begin{itemize}
    \item $t$ could be a tour that picks up tokens from nodes at bag $s$ and does not visit or pick up tokens in $C_{s_1} \cup C_{s_{2}}$. 
    \item $t$ could be a tour that picks up tokens from nodes at bag $s$ and picks up tokens only from $C_{s_{1}}$. 
    \item $t$ could be a tour that picks up tokens from nodes at bag $s$ and picks up tokens only from $C_{s_{2}}$. 
    \item $t$ could be a tour that picks up tokens from nodes at bag $s$ and picks up tokens from $C_{s_1} \cup C_{s_{2}}$. 
\end{itemize} 
We would find the minimum cost over all configurations $\yvec_s, \yvec', \yvec''$ as long as $\yvec_s, \yvec',\yvec''$ are consistent. We say $\yvec_s, \yvec', \yvec''$ are consistent if there is a way to write each tour in $\yvec_s$ as a combination of at most one tour from $\yvec'$, at most one tour from $\yvec''$ while also picking up extra tokens from nodes in $s$. We would also require that all tokens in $\yvec'$ and $\yvec''$ are picked up by tours in $\yvec_s$.

For a leaf bag $s$, $\ibold[\ovec_s, \yvec_s, \vec{0}, \vec{0}]$ denotes the minimum cost of tours entering bag $s$ and visiting the nodes in $s$ such that all tokens in $s$ are picked up by some tour in $\yvec_s$. The last two entries are set to $\vec{0}$ since $s$ is a leaf bag, and has no children, and there are no other tours (apart from those in $\yvec_s$) entering or exiting through nodes in bag $s$. We will set $\abold[s, \yvec_s] = \ibold[\ovec_s, \yvec_s, \vec{0}, \vec{0}]$ since $\ibold[\ovec_s, \yvec_s, \vec{0}, \vec{0}]$ computes exactly the minimum cost collection of partial tours covering $C_s = s$ having tour profiles corresponding to $\yvec_s$. We will explain how to compute the entries of $\ibold[\cdot, \cdot, \cdot, \cdot]$ in the next section. 

\subsection{Checking Consistency}
In our dynamic program, we are given three vectors $\yvec_s, \yvec', \yvec''$ where $s$ is a bag having child bags $s_1$ and $s_2$.  $\yvec'$ represents the configuration of tours covering $C_{s_1}$ and $\yvec''$ represents the configuration of tours covering $C_{s_2}$. Given a $\yvec_s$, for each node $u$ in $s$, there are $\ovec_{s,u}$ many tokens to be picked up at $u$. We require the tokens for nodes in $s$ and tokens covered by the partial tours from $\yvec'$ and $\yvec''$ to be picked up by tours in $\yvec_s$. 
For simplicity, we will refer to a tour from $\yvec_s$ as $t_s$, $\yvec'$ as $t_u$ and a tour from $\yvec''$ as $t_w$. 
\begin{definition}
We say configurations $\yvec_s, \yvec'$ and $\yvec''$ are \textbf{consistent} if the following holds: 
\begin{itemize}
    \item Every tour in $\yvec'$ maps to some tour in $\yvec_s$. 
    \item Every tour in $\yvec''$ maps to some tour in $\yvec_s$. 
    \item Every tour in $\yvec_s$ has at most two mapping to it and both cannot be from $\yvec'$ or $\yvec''$. 
    \item Suppose only one tour $t_u$ ($t_w$) maps to a tour $t_s$ in $\yvec_s$. The number of extra tokens (from nodes in $s$) in total picked up by tour $t_s$ from nodes in bag $s$ is exactly $|t_s| - |t_u|$ ($|t_s| - |t_w|$).
    \item Suppose $t_s$ has two tours: $t_u$ in $\yvec'$ and $t_w$ in $\yvec''$ mapping to it, then the number of extra tokens (from nodes in $s$) picked up by tour $t_s$ at $s$ is exactly $|t_s| - |t_u| - |t_w|$. 
    \item All tokens of nodes at bag $s$, $\ovec_s$ are picked up tours in $\yvec_s$. 
\end{itemize}
\end{definition}

Consistency ensures that we can patch up tours from subproblems and combine them into new tours in a correct manner while also picking up extra tokens from nodes in $s$. We will describe how we can compute consistency. Instead of using $\yvec_s$, we will use $\zvec_s$ which is the same as $\yvec_s$, but excludes information about the number of tokens in a bag, and only tracks information about the number of tours passing through bag $s$. 
$$\zvec_s = [\pvec_{s,v_1,v_1}, \pvec_{s,v_1,v_2}, \ldots, \pvec_{s,v_{d}, v_{d-1}}, \pvec_{s, v_d, v_d}]. $$
We will similarly define $\zvec'$ and $\zvec''$. Suppose $t_{s,x_1,x_2}$ is a tour in $s$ which enters through $x_1$ and exits through $x_2$, let $\zvec_s - t_{s,x_1,x_2}$ refers to the configuration $\zvec_s$ having one less tour of size $|t_{s,x_1,x_2}|$ from tours entering through $x_1$ and exiting through $x_2$ in $s$. Recall that $\ovec_s$ is the vector of extra tokens at each node in bag $s$ which need to be covered by tours in $\zvec_s$. 

Given $\zvec_s, \zvec', \zvec''$ and $\ovec_s$, we will use the table $\hbold$ to check if $\zvec_s, \zvec', \zvec''$ are consistent.  Let $\hbold[\ovec_s, \zvec_s, \zvec', \zvec''] = $True if $\zvec_s, \zvec'$ and $\zvec''$ are consistent and False otherwise. For the base case, $\hbold[\vec{0},\vec{0},\vec{0},\vec{0}] = $True. For the recurrence, we will look at all possible ways of combining tours from $\zvec'$ and $\zvec''$ into $\zvec_s$ while also picking up extra tokens from bag $s$. For a tour $t_s$, let $\ovec'_{s,t_s}$ be a vector where $\ovec'_{s,t_s,u}$ denotes the number of extra tokens picked up by $t_s$ at node $u$ in bag $s$. Let $\norm{\ovec'_{s,t_s}} = \sum_{u \in s}\ovec'_{s,t_s,u}$ count the number of tokens picked up by $t_s$ from nodes in $s$. 

Recall that a tour $t_s$ merges with at most one tour $t_u$ from $\zvec'$ and at most one tour $t_w$ from $\zvec''$. Similar to the case of trees, we can write the recurrence of our consistency table as: 
$$\hbold[\ovec_s, \zvec_s, \zvec', \zvec''] = \underset{\substack{t_s, t_u, t_w, \ovec'_{s,t_s} \\ |t_s| = |t_u| + |t_w| + \norm{\ovec'_{s,t_s}}}}{\bigvee} \hbold[\ovec_s - \ovec'_{s,t_s}, \zvec_s - t_s, \zvec' - t_u, \zvec'' - t_w].$$

Although the above DP lets us check if $\yvec_s, \yvec'$ and $\yvec''$ are consistent, the entries of $\hbold$ are True/False and does not give us information about the optimum order in which tour $t_s$ should visit nodes in $s$ or the cost associated with such an ordering. Suppose the tour $t_s$ visited $k_s$ nodes in bag $s$, there are $O(k^{k_s})$ many paths that tour $t_s$ can choose to take and each path has a cost associated with it. Our goal is to find a path having the smallest cost while also picking up tokens from nodes in bag $s$. We will next compute the minimum cost way to visit nodes in $s$ and pick up tokens from them. Recall the recurrence of our dynamic program for $\abold$ is the following, 
$$ \abold[s, \yvec_s] = \mi{\yvec', \yvec'' :\hbold[\ovec_s, \yvec_s, \yvec', \yvec''] = \text{True}}\{ \abold[s_1, \yvec'] +\abold[s_2, \yvec''] + \ibold[\ovec_s, \yvec_s, \yvec', \yvec''] \}.$$
The cost of using edges in $C_{s_1}$ and $C_{s_2}$ by the partial tours in $\yvec'$ and $\yvec''$ in $\yvec_s$ are accounted for by $ \abold[s_1, \yvec'] +\abold[s_2, \yvec'']$. However, we have not accounted for the cost of hopping from one node to the other in $s$ and also the cost of going from nodes in $s$ to nodes in child bags, $s_1$ and $s_2$. Note that a tour $t_s$ enters and exits through each node in $s$ at most once. Note that a tour visits a node $u$ in $s$ if it either has to pick up tokens at $u$ or if it uses $u$ to enter the child bag. If a tour $t_s$ enters and exist a node two or more times, we can short cut it so that it enters and exits only once. A tour in $t_s$ can visit up to $k$ nodes in a bag $s$ and it could use one of the nodes to enter a child bag ($s_1$ or $s_2$) and if so, it would use a node in $s$ to return to the bag $s$. This means the tour $t_s$ could visit up to $k$ nodes in $s$. Let $P_{t_s}$ be the ordered collection of edges where either both endpoints are in $s$ or one endpoint is in $s$ and the other is in $s_1 \cup s_2$. There are $O((3k)^{3k})$ possible permutations of for $P_{t_s}$ and $t_s$ could pick up at most $Q$ tokens from each node that it visits and each permutation has an associated cost with it. The number of possibilities for $P_{t_s}$ characterized by the the order of visiting nodes and the number of tokens picked up by tour $t_s$ from the $k$ nodes in bag $s$ is at most $O(Q^k (3k)^{3k})$. We will let $\cost(P_{t_s})$ denote the cost of the edges in $P_{t_s}$. The following figure illustrates an example of one such tour $t_s$ (in red) and $P_{t_s}$ (in blue). 
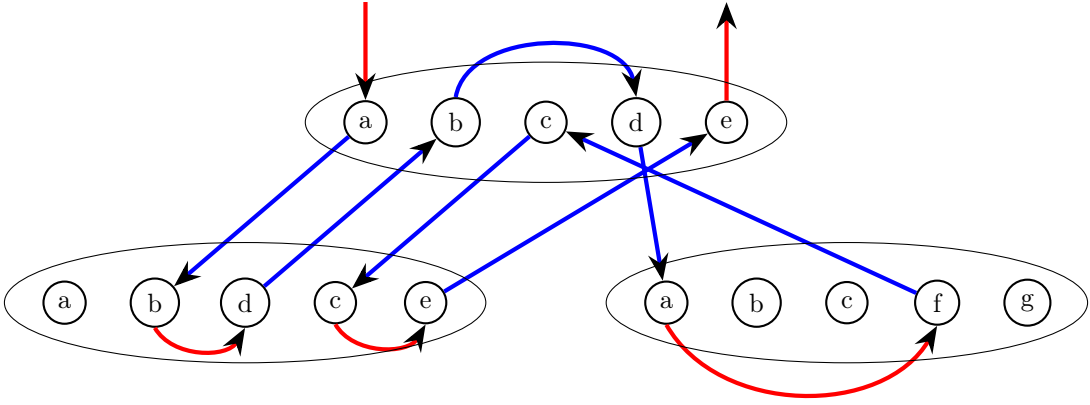
\begin{figure}[H]
\centering
\begin{tikzpicture}[scale=0.8]
\begin{scope}[every node/.style={circle,thick,draw}]
    \node (A_1) at (2,3) {a};
    \node (B_1) at (3.5,3) {b};
    \node (C_1) at (5,3) {c};
    \node (D_1) at (6.5,3) {d};
    \node (E_1) at (8,3) {e};
    
    \node (A_2) at (-3,0) {a};
    \node (B_2) at (-1.5,0) {b};
    \node (C_2) at (1.5,0) {c};
    \node (D_2) at (0,0) {d};
    \node (E_2) at (3,0) {e};
    
    \node (A_3) at (7,0) {a};
    \node (B_3) at (8.5,0) {b};
    \node (C_3) at (10,0) {c};
    \node (F_3) at (11.5,0) {f};
    \node (G_3) at (13,0) {g};
\end{scope}

\begin{scope}[>={Stealth[black]},
              every node/.style={fill=white,circle},
              every edge/.style={}]
              
    \draw[ultra thick, red, ->] (2,5)-- (A_1.north);
    \draw[ultra thick, blue,->] (A_1) -- (B_2);
    \draw[ultra thick, red,->,bend right=60] (B_2.south) to (D_2.south);
    \draw[ultra thick, blue,->] (D_2) -- (B_1);
    \draw[ultra thick, blue,->,bend right=280] (B_1.north) to (D_1.north);
    \draw[ultra thick, blue,->] (D_1) -- (A_3);
    \draw[ultra thick, red,->,bend right=60] (A_3.south) to (F_3.south);
    \draw[ultra thick, blue,->] (F_3) -- (C_1);
    \draw[ultra thick, blue,->] (C_1) -- (C_2);
    \draw[ultra thick, red,->,bend right=60] (C_2.south) to (E_2.south);
    \draw[ultra thick, blue,->] (E_2) -- (E_1);
    \draw[ultra thick, red, ->] (E_1.north)-- (8,5);
\end{scope}

\draw \boundellipse{0,0}{4}{1};
\draw \boundellipse{10,0}{4}{1};
\draw \boundellipse{5,3}{4}{1};

\end{tikzpicture}
\caption{Blue edges represent one such edge set for a particular tour $t_s$}
\end{figure}
Although $\hbold$ tells us if $\yvec_s, \yvec'$ and $\yvec''$ are consistent, $\hbold[\ovec_s, \zvec_s, \zvec', \zvec'']$ does not give us the cost of patching up $\yvec'$ and $\yvec''$ to form $\yvec_s$. We will use $\hbold$ to compute $\ibold$. Let $\ibold[\ovec_s, \zvec_s, \zvec', \zvec'']$ denote the cost of using the edges in bag $s$, edges connecting nodes in $s$ and $s_1$, and edges connecting nodes in $s$ and $s_2$. We can think of $\ibold$ as the cost of using edges to patch up partial tours covering $C_{s_1}$, $\zvec'$ and partial tours covering $C_{s_2}$,
$\zvec''$, to create tours covering $C_s$, $\zvec_s$. For the base case, we will set $\ibold[\vec{0},\vec{0},\vec{0},\vec{0}] =0$ and set all other entries to infinity. We will only compute an entry $\ibold[\ovec_s, \zvec_s, \zvec', \zvec'']$ if $\hbold[\ovec_s, \zvec_s, \zvec', \zvec'']=$True. Along with all possible values of $\ovec'_s, t_s,t_u,t_w$, we will also look at all possible paths $P_{t_s}$. In our recurrence, we are taking a tour $t_s$ from $\yvec_s$ along with maybe a tour $t_u$ from $\yvec'$, maybe a tour $t_w$ from $\yvec''$ along with tokens $\ovec'_s$ that $t_s$ covers at nodes in bag $s$. For such a tour $t_s$, there are $O(Q^k (3k)^{3k})$ many possibilities for $P_{t_s}$. For a fixed $P_{t_s}$, $\cost(P_{t_s})$ is the cost of forming $t_s$ from patching up $t_u$ and $t_w$ while picking up extra tokens from nodes in $s$. We will enumerate through all possibilities, break the recurrence into subproblems and find a solution of minimum cost. We can write the recurrence as follows: 
$$\ibold[\ovec_s, \zvec_s, \zvec', \zvec''] = \underset{\substack{t_s, t_u, t_w, P_{t_s}, \ovec'_{s,t_s} \\ |t_s| = |t_u| + |t_w| + \norm{ \ovec'_{s,t_s}}}}{\text{min}}\left\{ \cost(P_{t_s}) + \ibold[\ovec_s - \ovec'_{s,t_s}, \zvec_s - t_s, \zvec' - t_u, \zvec'' - t_w] \right\}.$$

\subsection{Time Complexity}
We will work bottom-up and analyze the time complexity of $\abold[\cdot, \cdot]$ on the assumption that we have already precomputed our consistency table $\ibold[\cdot, \cdot, \cdot, \cdot]$. Computing $\abold[s, \cdot]$ requires looking at entries of child bags in $\abold[\cdot, \cdot]$. Given $\yvec_s, \yvec'$ and $\yvec''$ which are consistent, computing the cost of $\abold[s, \yvec_s]$ takes $O(1)$ time. Each $\yvec_s$ consists of $O(k^2)$ different $\pvec_{s,u,v}$ vectors. Each $\pvec_{s,u,v}$ contains $\tau$ many triples $(\tvec^{s,x,z,i},\hvec^{s,x,z,i},\lvec^{s,x,z,i})$. 
\begin{enumerate}
    \item Each $\tvec^{s,x,z,i}$ has $n^{O(\log^2n/\eps)}$ possibilities since there are at most $O(\log^2 n/\eps)$ tours in a small bucket. 
    \item Each $\hvec^{s,x,z,i}$ and $\lvec^{s,x,z,i}$ have $n^{O(g)}$ possibilities. Recall that $g = (2 \delta \log n)/\eps$, so each $\hvec^{s,x,z,i}$ and $\lvec^{s,x,z,i}$ have $n^{O(\log n /\eps)}$ possibilities. 
    \item Each triple $(\tvec^{s,x,z,i},\hvec^{s,x,z,i},\lvec^{s,x,z,i})$ has $n^{O(\log^2 n/\eps)}$  possibilities. 
    \item Since $\pvec_{s,u,v}$ has $\tau = O(\log Q/ \eps)$ many such triples, the number of possible entries for $\pvec_{s,u,v}$ is $n^{O(\tau\log^2 n/\eps)} = n^{O(\log Q\log^2 n/\eps^2)}$. 
    \item Since $\yvec_s$ consists of $O(k^2)$ different entries of $\pvec$, so the total number of possible entries for each $\yvec_s$ is $n^{O(k^2\log Q\log^2 n/\eps^2)}$.
\end{enumerate}
Since there are $n^{O(k^2\log Q\log^2 n/\eps^2)}$ possibilities for $\yvec_s, \yvec'$ and $\yvec''$, the time of computing DP entries of $\abold[s, \cdot]$ for a single bag $s$ would take $n^{O(k^2\log Q\log^2 n/\eps^2)}$ and across all bags of the tree decomposition, it would still be $n^{O(k^2\log Q\log^2 n/\eps^2)}$. 

Now, we will analyze the time of computing the consistency table $\ibold[\cdot, \cdot, \cdot, \cdot]$. Assuming we have computed smaller entries, the cost of computing if $\ibold[\ovec_s, \zvec_s,\zvec', \zvec'']$ requires taking all possibilities way of picking $t_s, t_u, t_w, \ovec'_s, P_{t_s}$. Since there are at most $O(n)$ different tours, the number of possible ways of picking $t_s, t_u$ and $t_w$ is $O(n^3)$. Since the number of entries in the vector of $\ovec$ is $O(k)$, there are $n^{O(k)}$ possibilities for $\ovec'_s$. Each path $P_{t_s}$ consists of $O(k)$ nodes and at most $Q$ tokens can be picked up from each node, this would lead to $O(Q^k (3k)^{3k})= (nk)^{O(k)}$ many possibilities for $P_{t_s}$ since $Q \le n$. Hence, the total cost of computing a single entry of $\ibold[\cdot, \cdot, \cdot, \cdot]$ is $(nk)^{O(k)}$. Similar to the analysis for $\abold[\cdot, \cdot]$, there are $n^{O(k^2\log Q\log^2 n/\eps^2)}$ possibilities for $\ovec_s, \zvec_s,\zvec', \zvec''$, hence the total cost of computing $\ibold[\cdot, \cdot, \cdot, \cdot]$ is $(nk)^{O(k)}n^{O(k^2\log Q\log^2 n/\eps^2)}$. Similarly, the cost of computing $\hbold[\cdot, \cdot, \cdot, \cdot]$ is $(nk)^{O(k)}n^{O(k^2\log Q\log^2 n/\eps^2)}$.

Since the cost of computing $\ibold[\cdot, \cdot, \cdot, \cdot]$ dominates the cost of computing $\abold[\cdot, \cdot]$, the total time complexity of our algorithm is $(nk)^{O(k)}n^{O(k^2\log Q\log^2 n/\eps^2)} = n^{O(k^2\log Q\log^2 n/\eps^2)}$. Hence, for the unit demand case, since $Q \le n$, the runtime of our algorithm is $n^{O(k^2\log^3 n/\eps^2)}$.  

\subsection{Extension to Splittable CVRP in Bounded Treewidth Graphs}
We will extend our algorithm for unit demand CVRP on bounded-treewidth graphs to the splittable CVRP when demands are quasi-polynomially bounded.  In our algorithm for unit demand CVRP for bounded-treewidth CVRP, we viewed the unit demand of each node as a token placed at the node. For the splittable case, we can rescale the demand $d(v)$ such that there are $1 \le d(v) < nQ$ tokens on a node and we can use the same structure theorem as before by modifying tours such that there are at most $O(\log Q \log^2 n/\eps^2)$ different tours for partial tours at a node. We can use the same DP to compute the solution. Each $\yvec_s$ consists of $O(k^2)$ different $\pvec_{s,u,v}$ vectors. Each $\pvec_{s,u,v}$ contains $\tau$ many triples $(\tvec^{s,x,z,i},\hvec^{s,x,z,i},\lvec^{s,x,z,i})$. 
\begin{enumerate}
    \item Each $\tvec^{s,x,z,i}$ has $(nQ)^{O(\log^2n/\eps^2)}$ possibilities since there are at most $O(\log^2 n/\eps)$ tours in a small bucket. 
    \item Each $\hvec^{s,x,z,i}$ and $\lvec^{s,x,z,i}$ have $(nQ)^{O(g)}$ possibilities. Recall that $g = (2 \delta \log n)/\eps^2$, so each $\hvec^{s,x,z,i}$ and $\lvec^{s,x,z,i}$ have $(nQ)^{O(\log n /\eps^2)}$ possibilities. 
    \item Each triple $(\tvec^{s,x,z,i},\hvec^{s,x,z,i},\lvec^{s,x,z,i})$ has $(nQ)^{O(\log^2 n/\eps)}$  possibilities. 
    \item Since $\pvec_{s,u,v}$ has $\tau = O(\log Q/ \eps)$ many such triples, the number of possible entries for $\pvec_{s,u,v}$ is $(nQ)^{O(\tau\log^2 n/\eps)} = (nQ)^{O(\log Q\log^2 n/\eps^2)}$. 
    \item Since $\yvec_s$ consists of $O(k^2)$ different entries of $\pvec$, the total number of possible entries for each $\yvec_s$ is $(nQ)^{O(k^2\log Q\log^2 n/\eps^2)}$.
\end{enumerate}
Similar to the analysis of the runtime of the unit demand case, the time complexity of computing the entries of DP tables $\abold$ and consistency table $\ibold$ is, $(kQ)^{O(k)}(nQ)^{O(k^2\log Q\log^2 n/\eps^2)} = (nQ)^{O(k^2\log Q\log^2 n/\eps^2)}$ since $k \le n$. Suppose $Q = n^{O(\log^c n)}$, then the runtime of our algorithm is $n^{O(k^2 \log^{2c + 3}n/\eps^2)}$. 

\section{Extension to Splittable CVRP for Graphs of Bounded Doubling Metrics and Bounded Highway Dimension}
In this section, we will show how we can use our algorithm for CVRP on bounded-treewidth graphs as a blackbox to obtain a QPTAS for graphs of bounded doubling metrics and graphs of bounded highway dimension.  We will use the following result about emdedding graphs of doubling dimension $D$ into a bounded-treewidth graph of treewidth $k \le 2^{O(D)}\ceil*{\br{\frac{4D\log \Delta}{\eps}}^D}$ by Talwar \cite{Talwar-embedding}. 
\begin{lemma}\label{lem:doubling-embed}
(Theorem 9 in \cite{Talwar-embedding}) Let $(X,d)$ be a metric with doubling dimension $D$ and aspect ratio $\Delta$. For any $\eps > 0$, $(X,d)$ can be $(1 + \eps)$ probabilistically approximated by a family of treewidth $k$-metrics for $k \le 2^{O(D)}\ceil*{\br{\frac{4D\log \Delta}{\eps}}^D}$.
\end{lemma}
We will also use the following result by Feldmann et al. \cite{FeldmannFKP15-embedding} related to graphs of low highway dimension. 
\begin{lemma}\label{lem:highway-embe}
(Theorem 3 in \cite{FeldmannFKP15-embedding}) Let $G$ be a graph with highway dimension $D$ of violation $\lambda > 0$, and aspect ratio $\Delta$. For any $\eps > 0$, there is a polynomial-time computable probabilistic embedding $H$ of $G$ with treewidth $(\log \Delta)^{O\br{\log^2(\frac{D}{\eps \lambda})/\lambda}}$ and expected distortion $1 + \eps$. \end{lemma}

For both graph classes, our algorithm works as follows. The input graph $G$ is embedded into a host graph $H$ of bounded treewidth using the embedding given in Lemma \ref{lem:doubling-embed} and Lemma \ref{lem:highway-embe}. The algorithm then finds a $(1 + \eps)$-approximation for CVRP for $H$, using the dynamic programming solution from the Section 5. The solution for $H$ is then \emph{lifted} back to a solution in $G$. For each tour in the solution for $H$, a tour in $G$ will visit nodes in the same order as the tour in $H$. The embedding given in Lemma \ref{lem:doubling-embed} and Lemma \ref{lem:highway-embe} is such that an optimal set of tours in the host graph gives a $(1 + \eps)$ solution in $G$. The embedding also ensures that $H$ has treewidth small enough that the algorithm runs in quasi-polynomial time. 
\begin{theorem}
For any $\eps > 0$ and $D > 0$, there is a an algorithm that, given an instance of the splittable CVRP with capacity $Q = n^{\log^c n}$ and the graph has doubling dimension $D$ with cost $\opt$, finds a  
$(1 + \eps)$-approximate solution in time 
$n^{O(D^D \log^{2c + D + 3}n/\eps^{D+2})}$.
\end{theorem}
\begin{proof}
This follows easily from Lemma \ref{lem:doubling-embed} and using the algorithm for bounded-treewidth as a blackbox. In place of $k$, we will substitute $k = 2^{O(D)}\ceil*{\br{\frac{4D\log \Delta}{\eps}}^D}$ into the runtime for the algorithm for bounded-treewidth which is $n^{O(k^2 \log^{2c + 3}n/\eps^2)}$. Hence, we have an algorithm for graphs of bounded doubling dimension with runtime $n^{O(D^D \log^{2c + D + 3}n/\eps^{D+2})}$. 
\end{proof}

As an immediate corollary, since $\mathbb{R}^2$ has doubling dimension 7 \cite{disk-covering}, the above theorem implies an approximation scheme for unit demand CVRP on Euclidean metrics on $\mathbb{R}^2$ in time $n^{O(\log^{10}n/\eps^{9})}$ which improves on the run time of $n^{\log^{O(1/\epsilon)}n}$ of \cite{Das-Mathieu}. 
\begin{theorem}
For any $\eps > 0, \lambda > 0$ and $D > 0$, there is a an algorithm that, given an instance of the splittable CVRP with capacity $Q = n^{\log^c n}$ and a graph with highway dimension $D$ and violation $\lambda$ finds
a $(1+\eps)$-approximate solution in time $n^{O( \log^{2c + 3 + \log^2(\frac{D}{\eps \lambda})\cdot \frac{1}{\lambda}}n/\eps^2)}$.
\end{theorem}
\begin{proof}
This follows easily from Lemma \ref{lem:highway-embe} and using the algorithm for bounded-treewidth as a blackbox. In place of $k$, we will substitute $k = (\log \Delta)^{O\br{\log^2(\frac{D}{\eps \lambda})/\lambda}}$ into the runtime for the algorithm for bounded-treewidth which is $n^{O(k^2 \log^{2c + 3}n/\eps^2)}$. Hence, we have an algorithm for graphs of bounded doubling dimension with runtime $n^{O( \log^{2c + 3 + \log^2(\frac{D}{\eps \lambda})\cdot \frac{1}{\lambda}}n/\eps^2)}$. 
\end{proof}

\section{Conclusion}
In this paper we presented QPTAS's for CVRP on trees, graphs of bounded treewidths, bounded doubling dimension, and bounded highway dimension. The immediate questions to consider are whether these approximation schemes can in fact be turned into PTAS's. Even for the case of trees, although we can improve the run time slightly by shaving off one (or maybe two) log factors from the exponent, it is not clear if it can be turned into a PTAS without significant new ideas. 

Although our result implies a QPTAS with a better run time
for CVRP on Euclidean plan $\mathbb{R}^2$ ($n^{O(\log^{10}n/\eps^{9})}$ vs the time of
f $n^{\log^{O(1/\epsilon)}n}$ of \cite{Das-Mathieu}), getting a PTAS remains an interesting open question.
As discussed in \cite{AdamaszekCL09}, the difficult case appears to be when $Q$ is polynomial in $n$ (e.g. $Q=\sqrt{n}$). Another interesting question is to consider CVRP on planar graphs and develop approximation schemes for them and more generally graphs of bounded genus or minor free graphs.

\bibliographystyle{abbrv}
\bibliography{references}{}


\end{document}